%% file: main.tex
\documentclass[11pt]{article}
\input{macros}

\author{
	Anand Louis\footnote{Indian Institute of Science, Bangalore, India.}\\ \href{mailto:anandl@iisc.ac.in}{anandl@iisc.ac.in}
	\and 
	Rameesh Paul\footnotemark[1]\\ \href{mailto:rameeshpaul@iisc.ac.in}{rameeshpaul@iisc.ac.in} 
	\and 
	Arka Ray\footnotemark[1]\\ \href{mailto:arkaray@iisc.ac.in}{arkaray@iisc.ac.in}
}

\title{Sparse Cuts in Hypergraphs from Random Walks on Simplicial Complexes}

\date{}

\begin{document}     
\maketitle

\begin{abstract}
\input{abstract_text}
\end{abstract}
\newpage

\input{introduction}
\input{prelim}
\input{cheeger}
\input{counter-example}

\paragraph{Acknowledgements.}
We would like to thank Madhur Tulsiani and Fernando Jeronimo for their helpful comments on earlier versions of this paper.
RP was supported by Prime Minister's Research Fellowship, India.

\bibliographystyle{amsalpha}
\bibliography{references}
    
\appendix
\input{appendix}

\end{document}

%% file: macros.tex
\usepackage{tcolorbox}
\usepackage[T1]{fontenc}
\usepackage{mathpazo}
\usepackage{graphicx}
\usepackage{dirtytalk}
\usepackage{bbold}
\usepackage{caption}
\DeclareCaptionLabelFormat{nolabel}{}
\usepackage{chngcntr} 
\usepackage{enumerate} 
\usepackage{geometry} 
\usepackage{amsmath} 
\usepackage{amssymb} 
\usepackage{textcomp}
    \AtBeginDocument{%
        
    }
\usepackage{upquote} 
\usepackage{eurosym} 
\usepackage[mathletters]{ucs} 
\usepackage[utf8x]{inputenc} 
\usepackage{fancyvrb} 
\usepackage{grffile} 
\usepackage{longtable} 
\usepackage{booktabs} 
\usepackage[inline]{enumitem} 
\usepackage[normalem]{ulem} 

\usepackage[pdftex,pagebackref,colorlinks,linkcolor=blue,filecolor = blue, citecolor = magenta, urlcolor  = JungleGreen]{hyperref}

    \definecolor{urlcolor}{rgb}{0,.145,.698}
    \definecolor{linkcolor}{rgb}{.71,0.21,0.01}
    \definecolor{citecolor}{rgb}{.12,.54,.11}

    \definecolor{ansi-black}{HTML}{3E424D}
    \definecolor{ansi-black-intense}{HTML}{282C36}
    \definecolor{ansi-red}{HTML}{E75C58}
    \definecolor{ansi-red-intense}{HTML}{B22B31}
    \definecolor{ansi-green}{HTML}{00A250}
    \definecolor{ansi-green-intense}{HTML}{007427}
    \definecolor{ansi-yellow}{HTML}{DDB62B}
    \definecolor{ansi-yellow-intense}{HTML}{B27D12}
    \definecolor{ansi-blue}{HTML}{208FFB}
    \definecolor{ansi-blue-intense}{HTML}{0065CA}
    \definecolor{ansi-magenta}{HTML}{D160C4}
    \definecolor{ansi-magenta-intense}{HTML}{A03196}
    \definecolor{ansi-cyan}{HTML}{60C6C8}
    \definecolor{ansi-cyan-intense}{HTML}{258F8F}
    \definecolor{ansi-white}{HTML}{C5C1B4}
    \definecolor{ansi-white-intense}{HTML}{A1A6B2}

    \DefineVerbatimEnvironment{Highlighting}{Verbatim}{commandchars=\\\{\}}
    \definecolor{incolor}{rgb}{0.0, 0.0, 0.5}
    \definecolor{outcolor}{rgb}{0.545, 0.0, 0.0}

    \hypersetup{
      breaklinks=true,  % so long urls are correctly broken across lines
      colorlinks=true,
      urlcolor=urlcolor,
      linkcolor=linkcolor,
      citecolor=citecolor,
      }
    \usepackage{pstricks}

\usepackage{amsmath,amstext,amssymb,amsfonts}
\usepackage{graphicx}
\usepackage{xcolor}
\usepackage{tabularx}
\usepackage{verbatim}
\usepackage{fullpage}
\usepackage[T1]{fontenc}
\usepackage{bbm}

\usepackage{amsthm}
\usepackage{ifdraft}
\usepackage{algorithm2e}
\usepackage{algorithmic}
\usepackage{appendix}
\usepackage{multicol}

\usepackage{nicefrac}

\usepackage{microtype}

\newtheorem{theorem}{Theorem}
\newtheorem*{theorem*}{Theorem}

\newtheorem*{claim*}{Claim}

\newtheorem*{proposition*}{Proposition}
\newtheorem{lemma}[theorem]{Lemma}
\newtheorem*{lemma*}{Lemma}
\newtheorem{corollary}[theorem]{Corollary}

\newtheorem*{conjecture*}{Conjecture}

\newtheorem{fact}[theorem]{Fact}
\newtheorem*{fact*}{Fact}

\newtheorem*{hypothesis*}{Hypothesis}

\theoremstyle{definition}
\newtheorem{definition}[theorem]{Definition}
\newtheorem{construction}[theorem]{Construction}

\newtheorem{remark}[theorem]{Remark}

\usepackage{prettyref}
\newcommand{\savehyperref}[2]{\texorpdfstring{\hyperref[#1]{#2}}{#2}}

\newrefformat{parta}{\savehyperref{#1}{a}}
\newrefformat{partb}{\savehyperref{#1}{b}}
\newrefformat{eq}{\savehyperref{#1}{\textup{(\ref*{#1})}}}
\newrefformat{lem}{\savehyperref{#1}{Lemma~\ref*{#1}}}
\newrefformat{def}{\savehyperref{#1}{Definition~\ref*{#1}}}
\newrefformat{thm}{\savehyperref{#1}{Theorem~\ref*{#1}}}
\newrefformat{cor}{\savehyperref{#1}{Corollary~\ref*{#1}}}
\newrefformat{cha}{\savehyperref{#1}{Chapter~\ref*{#1}}}
\newrefformat{sec}{\savehyperref{#1}{Section~\ref*{#1}}}
\newrefformat{app}{\savehyperref{#1}{Appendix~\ref*{#1}}}
\newrefformat{tab}{\savehyperref{#1}{Table~\ref*{#1}}}
\newrefformat{fig}{\savehyperref{#1}{Figure~\ref*{#1}}}
\newrefformat{hyp}{\savehyperref{#1}{Hypothesis~\ref*{#1}}}
\newrefformat{alg}{\savehyperref{#1}{Algorithm~\ref*{#1}}}
\newrefformat{sdp}{\savehyperref{#1}{SDP~\ref*{#1}}}
\newrefformat{qp}{\savehyperref{#1}{QP~\ref*{#1}}}
\newrefformat{vp}{\savehyperref{#1}{VP~\ref*{#1}}}
\newrefformat{lp}{\savehyperref{#1}{LP~\ref*{#1}}}
\newrefformat{rem}{\savehyperref{#1}{Remark~\ref*{#1}}}
\newrefformat{item}{\savehyperref{#1}{Item~\ref*{#1}}}
\newrefformat{step}{\savehyperref{#1}{step~\ref*{#1}}}
\newrefformat{conj}{\savehyperref{#1}{Conjecture~\ref*{#1}}}
\newrefformat{fact}{\savehyperref{#1}{Fact~\ref*{#1}}}
\newrefformat{prop}{\savehyperref{#1}{Proposition~\ref*{#1}}}
\newrefformat{claim}{\savehyperref{#1}{Claim~\ref*{#1}}}
\newrefformat{relax}{\savehyperref{#1}{Relaxation~\ref*{#1}}}
\newrefformat{red}{\savehyperref{#1}{Reduction~\ref*{#1}}}
\newrefformat{part}{\savehyperref{#1}{Part~\ref*{#1}}}
\newrefformat{prob}{\savehyperref{#1}{Problem~\ref*{#1}}}
\newrefformat{ass}{\savehyperref{#1}{Assumption~\ref*{#1}}}
\newrefformat{cons}{\savehyperref{#1}{Construction~\ref*{#1}}}
\newrefformat{obs}{\savehyperref{#1}{Observation~\ref*{#1}}}

\newcommand{\Sref}[1]{\hyperref[#1]{\S\ref*{#1}}}

\renewcommand{\leq}{\leqslant}

\renewcommand{\geq}{\geqslant}

\usepackage{bm}

\usepackage{xspace}

\usepackage{boxedminipage}

\newcommand{\mper}{\,.}

\newcommand{\paren}[1]{\left(#1 \right )}

\newcommand{\Brac}[1]{\left[#1\right]}

\newcommand{\set}[1]{\left\{#1\right\}}

\newcommand{\abs}[1]{\left\lvert#1\right\rvert}

\newcommand{\norm}[1]{\left\lVert#1\right\rVert}

\newcommand{\defeq}{\stackrel{\textup{def}}{=}}

\newcommand{\inprod}[1]{\left\langle #1\right\rangle}

\newcommand{\R}{\mathbb R}

\newcommand{\Esymb}{\mathbb{E}}
\newcommand{\Psymb}{\mathbb{P}}

\DeclareMathOperator*{\E}{\Esymb}

\DeclareMathOperator*{\ProbOp}{\Psymb}

\newcommand{\pr}[2]{\ProbOp_{#1}\Brac{#2}}

\newcommand{\e}{\epsilon}

\definecolor{DSgray}{cmyk}{0,0,0,0.7}

\let\e\varepsilon

\newcommand{\vol}{\mathsf{vol}}
\newcommand{\bigO}{\mathcal{O}}

\newcommand{\rank}{{\sf rank}}
\newcommand\numberthis{\addtocounter{equation}{1}\tag{\theequation}} %number specific equations

%% file: abstract_text.tex
There are a lot of recent works on generalizing the spectral theory of graphs and graph partitioning to $k$-uniform hypergraphs.
There have been two broad directions toward this goal.
One generalizes the notion of graph conductance to hypergraph conductance \cite{lm16, cltz18}.
In the second approach one can view a hypergraph as a simplicial complex and study its various topological properties \cite{LM06, MW09, DKW16, PR17} and spectral properties \cite{KM17, DK17, KO18a, KO18b, Opp18}.

In this work, we attempt to bridge these two directions of study by relating the spectrum of {\em up-down walks} and {\em swap-walks} on the simplicial complex to hypergraph expansion. In surprising contrast to random-walks on graphs, we show that the spectral gap of swap-walks and up-down walks between level $m$ and $l$ with $1 < m \leq l$ can not be used to infer any bounds on hypergraph conductance.
Moreover, we show that the spectral gap of swap-walks between $X(1)$ and $X(k-1)$ can not be used to infer any bounds on hypergraph conductance, whereas we give a Cheeger-like inequality relating the spectral of walks between level $1$ and $l$ for any $l \leq k$ to hypergraph expansion.
This is a surprising difference between swaps-walks and up-down walks!

Finally, we also give a construction to show that the well-studied notion of {\em link expansion} in simplicial complexes can not be used to bound hypergraph expansion in a Cheeger like manner.

%% file: introduction.tex
\section{Introduction}
\label{sec:introduction}
In recent years, there have been two broad directions of generalizations of graph partitioning and the spectral theory of graphs to hypergraphs.
One direction attempts to generalize the notion of conductance in graphs to hypergraphs \cite{lm16,cltz18}.
The graph expansion (also referred as graph conductance) is defined as

\begin{align*}
    &\phi_G \defeq \min_{\substack{S \subseteq V\\ \mathsf{
    vol}_G(S) \leq \frac{\mathsf{vol}_G(V)}{2}}} \phi_G(S),\text{ where } \phi(S) \defeq \frac{w({\partial_G(S)})}{\mathsf{vol}_G(S)}
\end{align*}
with $\mathsf{vol}_G(S)$ being the sum of degrees of the vertices in $S$ and $\partial_G(S)$ being the edges crossing the boundary of the set $S$, hence $w(\partial_G(S))$ is the sum of weights of the edges on the boundary.
Analogously, the hypergraph expansion/conductance is defined as
\begin{align*}
    &\phi_H \defeq \min_{\substack{S \subseteq V\\ \mathsf{
    vol}_H(S) \leq \frac{\mathsf{vol}_H(V)}{2}}} \phi_H(S),\text{ where } \phi_H(S) \defeq \frac{\Pi\paren{\partial_H(S)}}{\mathsf{vol}_H(S)}
\end{align*}
with $\mathsf{vol}_H(S)$ being the sum of degrees of the vertices in $S$, and $\partial_H(S)$ being the edges crossing the boundary of the set $S$, and $\Pi(\partial_H(S))$ is the sum of weight of edges on the boundary.

Another direction views a hypergraph as a {\em simplicial complex} and studies its various topological properties \cite{LM06,MW09,DKW16,PR17} and spectral properties \cite{KM17,DK17,KO18a,KO18b,Opp18}.
The work \cite{DK17} introduced a generalization of random-walks on graphs to random-walks over the faces of the simplicial complex; this random-walk has found numerous applications in a myriad of other problems \cite{DK17,DD19,ALGV19,ALG20,ALG22}, etc., to state a few.

There has been a lot of work on understanding the relationship between random-walks on graphs (including the spectra of the random-walk operator) and graph partitioning.
The celebrated Cheeger's inequality gives one such relation between the graph expansion and the second eigenvalue of the random-walk matrix $\lambda_2$ as,  
\begin{align*}
    \frac{1-\lambda_2}{2} \leq \phi_G \leq \sqrt{2(1-\lambda_2)}.
\end{align*}

In this work, we aim to bridge the gap between these two directions by studying the relationship between hypergraph expansion and random-walks on the corresponding simplicial complex.

In a seminal work, \cite{ABS10} showed that if a graph has a \say{small} threshold-rank\footnote{the number of ``large'' eigenvalues of the adjacency matrix, see \prettyref{def:threshold_rank} for formal definition.}, then they can compute a near optimal assignment to unique games in time exponential in the threshold rank.
The works \cite{BRS11,GS11} gave an SoS hierarchy based algorithm generalizing this result to any 2-CSP.
The work \cite{AJT19} introduces the notion of {\em swap-walks} and use that to define a notion of threshold-rank for simplicial complexes.
Using their notion of threshold-rank they generalized the results of \cite{BRS11,GS11} to $k$-CSPs.
Further, \cite{ABS10} showed that large threshold-rank graphs must have a small non-expanding set (they also gave a polynomial time algorithm to compute such a set). A natural open question from the work of \cite{AJT19,JST21} is whether hypergraphs with large threshold-rank (the hypergraph analogue is called non-splittability)
have a small non-expanding set?
Our first result answers this question negatively.
\begin{theorem}[Informal Version of \prettyref{thm:non_splittable} and \prettyref{cor:non_splittable}]
\label{thm:informal_non_splittable}
For any $n\geq 6,k\geq 3$, there exists a $k$-uniform hypergraph $H$ with at least $n$ vertices such that $\phi_H\geq \frac{1}{k}$ but for any $m, l$, if either $m,l\geq 2$ or $m=k-l$, the swap-walk from $X(m)$ to $X(l)$ has threshold rank at least $\Omega_k(n)$ (for any $\tau \in [-1,1]$ as choice of threshold). 
Moreover, $H$ is not $(\tau,\Omega_k(n))$-splittable for any $\tau \in [-1,1]$.

\end{theorem}

For a splittable hypergraph, there is some $l$, such that the swap-walk graph between $X(l)$ and $X(k-l)$ has low threshold-rank. Then it follows from 
\prettyref{thm:informal_non_splittable} that there are non-splittable expanding hypergraphs (see \prettyref{cor:non_splittable} for the precise statement).

\cite{AJT19,DD19} show that for a high dimensional expander (HDX)\footnote{For formal definition see \prettyref{def:local_hdx}.} the swap-walks indeed have a large spectral gap\footnote{For a linear operator $\mathsf A:V\to W$ where $V\ne W$ the spectral gap refers to $\sigma_1(A)-\sigma_2(A)$, while for a linear operator $B:V\to V$, it refers to $\lambda_1(A)-\lambda_2(A)$.}.
However, we are interested in the case when the hypergraph instance is not an HDX.
One recalls that for a non-expanding graph, Cheeger's inequality and Fiedler's algorithm allow us to compute a combinatorial sparse cut in the graph.
Similarly, we ask the question whether one can compute a sparse cut in the input hypergraph in this setting?

Unfortunately, in the light of \prettyref{thm:informal_non_splittable},
computing a sparse cut in the hypergraph when swap-walks (in the setting studied by \cite{AJT19,JST21}; see \prettyref{thm:non_splittable} for the precise statement) have a small spectral gap is generally not possible. 
This is in surprising contrast to the case of graphs where the swap-walk reduces to the usual random-walk, and the second largest eigenvalue of the random-walk matrix is related to graph expansion via the Cheeger's inequality.

Next, we investigate whether the spectral gap of the up-down walk introduced by \cite{DK17} can be related to hypergraph expansion.
More formally, we investigate whether the spectral gap of the up-down walk between levels $X(m)$ and $X(l)$ $(l > m)$ be related to the hypergraph expansion in Cheeger like manner. Here, the answer depends on $m$ and $l$. We first show that if $m \geq 2$, then no such relation is possible.

\begin{theorem}[Informal Version of \prettyref{thm:general_updown}]
For any positive integers $n,k$ with $n\geq 6,k\geq 3$, there exists an $k$-uniform hypergraph $H$ on at least $n$ vertices such that $\phi_H\geq \frac{1}{k}$
% with distribution $\Pi$ induced by uniform distribution on $X(k)$
and for all positive integers $2\leq m< l\leq k$ the threshold rank of the up-down walk matrix between levels $X(m)$ and $X(l)$ is at least $\Omega_k(n)$ (for any $\tau \in [-1,1]$ as choice of threshold).
\end{theorem}

Contrasting this, we show that if $m = 1$, then such a relationship is indeed possible.

\begin{theorem}[Informal Version of \prettyref{thm:sparse_cut}]
\label{thm:informal_sparse_cut}
Given a hypergraph, where the second largest eigenvalue of the up-down walk matrix (of simplicial complex induced by the hypergraph) between levels $X(1)$ and $X(l)$, for some $l \in [k]$ is $1-\varepsilon$ we have $\frac{\varepsilon}{k}\leq\phi_H\leq 4\sqrt{\varepsilon}$. 
Furthermore, there exists a polynomial time algorithm which when given such a hypergraph, outputs a set $S$ such that its expansion in the hypergraph $\phi_H(S)\leq 4\sqrt{\varepsilon}$. 
\end{theorem}

\prettyref{thm:informal_sparse_cut} and \prettyref{thm:informal_non_splittable} also show a surprising difference between up-down walks and swap-walks whereby we can compute sparse cut on the hypergraph using up-down walk from $X(1)$ to $X(l), l\in[k]$ using a Cheeger-like inequality, whereas it is not possible (in general) to compute a sparse cut by considering the spectrum of swap-walks from $X(1)$ to $X(k-1)$.

A yet another notion of spectral expansion called {\em link-expansion} of a simplicial complex has been studied recently in many works \cite{KM17,DK17,KO18a,KO18b,Opp18} having applications in \cite{DK17,DD19,ALGV19,ALG20,ALG22} (see \prettyref{def:local_hdx} for formal definition).
Our final result shows that there exists hypergraphs with large hypergraph expansion and arbitrarily small link expansion.
Therefore, hypergraph expansion can not be bounded by link-expansion in  Cheeger like manner.

\begin{theorem}[Informal Version of \prettyref{thm:non_hdx}]
Let $n,k$ be any positive integers such that $n\geq 3k$ and $k\geq 3$, there exists a $k$-uniform hypergraph $H$ on $n+k-2$ vertices such that the link-expansion of the induced simplicial complex $X$ is at most $\bigO(\frac{1}{n^2})$ and the expansion of $H$ is at least $\Omega_k(1)$.
\end{theorem}
To the best of our knowledge, this is the first construction to show this.

The work~\cite{lm16}~(see Remark 1.9) used an example similar in spirit to our constructions to show that another notion of expansion on simplicial complexes called co-boundary expansion is incomparable to the hypergraph expansion.
In particular, they constructed a class of $k$-uniform hypergraphs, each of which has co-boundary expansion (at dimension $k$) as 1 but contains hypergraphs that have essentially arbitrary hypergraph expansion.
Still, \cite{lm16} did not give an explicit example that shows a separation between hypergraph expansion and the notion of link-expansion or the spectral gap or the threshold rank of the random walks on a simplicial complex (i.e., up-down walk, swap-walk).

The $m$-dimensional co-boundary expansion may also seem related to the expansion of the up-walk from the level $m-1$ to $m$ as both of these consider the ratio of the number of $m$-dimensional faces containing a set of $m-1$-dimensional faces to the volume of the set with the only difference being how the volume is computed.
Yet, we do not know if such a relation actually exists.
One may similarly compare the expansion of the down-walk and the boundary expansion.
But still, Steenbergen, Klivian, and Mukherjee \cite{SKM14} and Gundert, and Wagner \cite{GW16} were able to show that for the $m$-dimensional {co-boundary expansion} no Cheeger-type inequality can be shown whereas such a relation is immediate from Cheeger's inequality in case of up-walk.
Nevertheless, \cite{SKM14} obtained (under some minor assumptions) an extension of Cheeger's inequality on the $m$-dimensional boundary expansion.
Finally, \cite{DDFH18} showed that the operator norm of the difference between up-down and down-up walks between two consecutive levels is within a $\bigO(k)$ factor of link-expansion while no such relation between up-Laplacian, down-Laplacian (see~\cite{SKM14} for definition) and link-expansion is known.

\subsection{Additional Related Works} 

The work \cite{cltz18} generalized the Laplacian of graph to hypergraphs by expressing the graph Laplacian in terms of a non-linear diffusion process.
They showed an analogue of Cheeger's inequality relating the expansion of the hypergraph to the second smallest eigenvalue of the Laplacian.
Yoshida~\cite{Yoshida19} introduced the notion of submodular transformations and extended the notions of degree, cut, expansion, and Laplacian to them.
They derived the Cheeger's inequality in this setting.
This generalizes the Cheeger's inequality on graphs and hypergraphs (as in \cite{cltz18}) while showing similar inequalities for entropy.

There are also several works exploring Cheeger-like inequalities for simplicial complexes.
Parzanchevski, Rosenthal, and Tessler \cite{PRT16} defined the notion of Cheeger constant $h(X)$ for a simplicial complex, a generalization of sparsity of a graph.
The quantity $h(X)$ is the minimum over all partitions of the vertex set $V$ into $k$ sets the fraction of $k$-dimensional faces present crossing the partition compared to the maximum possible $k$-dimensional faces crossing the partition. 
They also showed that for simplicial complex $X$ with a complete skeleton $h(X)\geq \lambda(X)$ where $\lambda(X)$ is the link-expansion of the simplicial complex.
Gundert and Svedl\'ak \cite{GS15} gave an extension of this result to any simplicial complex.
Very recently, Jost and Zhang~\cite{JZ23} extended the Cheeger-like inequality for {bipartiteness ratio}\footnote{The bipartiteness ratio of $G$ is defined as $\beta_G=\min_{S\subseteq V,L\sqcup R=S} \frac{2\partial(L)+2\partial(R)+\partial(S)}{\mathsf{vol}_G(S)}$.} on graphs due to Trevisan~\cite{Tre12} to a cohomology based definition of bipartiteness ratio for simplicial complexes.

In case of a HDX, Bafna, Hopkins, Kaufmann and Lovett~\cite{BHKL22} consider high-dimensional walks (a generalization of swap-walks and up-down walks) on a levels $i<k$.
They then relate the (non-)expansion of a link\footnote{\cite{BHKL22} uses a different (albeit related) notion of the link of a face $\sigma\in X(j)$. There the link of a face $\sigma$ is the set of level-$i$ faces containing $\sigma$.} of a level-$j$ face (with $j\leq i$) in the graph corresponding to the walk and level-$j$ {approximate eigenvalue} of the walk.
Here $\lambda_j$ is the level-$j$ approximate eigenvalue of a high-dimensional walk $\mathsf M$ if there is a function $f_j$ such that $\norm{\mathsf Mf_j-\lambda_jf_j}\leq O(\sqrt\gamma)\norm{f_j}$ and $f_i=\mathsf U^{i-j}g$ where $g\in \R^{X(j)}$.

%% file: prelim.tex
\subsection{Preliminaries}
\label{sec:prelims}

\subsubsection{Linear Algebra}
We recall a few facts and definitions from linear algebra.
\begin{fact}[\cite{HK71}]
Let $V, W$ be two vector spaces with inner products $\inprod{\cdot,\cdot}_V$, $\inprod{\cdot,\cdot}_W$.
If $A:V\to W$ be a linear operator then there exists a unique linear operator $B:W\to V$ such that $\inprod{\mathsf Af, g}_W=\inprod{f,\mathsf B g}_V$.
If $v\in V$ then there exists a unique linear operator $C:V\to \R$ such that $C u=\inprod{v,u}_V$ for any $u\in V$.
\end{fact}
\begin{definition}
\label{def:adj}
Given a linear operator $\mathsf A:V\to W$ between two vector spaces $V$ and $W$ with  inner products $\inprod{\cdot,\cdot}_V$ and $\inprod{\cdot,\cdot}_W$ defined on them, the adjoint of $\mathsf A$ is defined as the (unique) linear operator $\mathsf A^\dagger:W\to V$ such that
$\inprod{\mathsf Af, g}_W=\inprod{f,\mathsf A^\dagger g}_V$
for any $f\in V$ and $g\in W$.
Furthermore, given any $v\in V$ we define $v^\dagger:V\to \R$ as the linear operator which satisfies $v^\dagger u=\inprod{v,u}_V$ for any $u\in V$.
\end{definition}

It can be easily verified that most properties of the transpose of an operator also hold for the adjoint, e.g., $(\mathsf A^\dagger)^\dagger=\mathsf A$, $(\mathsf{AB})^\dagger= \mathsf B^\dagger \mathsf A^\dagger$, etc.

\begin{definition}
Given a linear operator $\mathsf A:V\to W$ between two inner product spaces $V$ and $W$ a singular value $\sigma$ is a non-negative real number such that there exists $v\in V$ and $w\in W$ which satisfy
$\mathsf Av=\sigma w$ and $w^\dagger \mathsf A=\sigma v^\dagger$.
The vectors $v$ and $w$ are called the right and left singular vectors, respectively, associated with the singular value $\sigma$.
We denote the $i$-th largest singular value of $\mathsf A$ by $\sigma_i(\mathsf A)$.
\end{definition}

\begin{fact}
\label{fact:svalue_eigenvalue}
Let $V$,$W$ be two inner product spaces and $\mathsf A:V\to W$ be a linear operator.
Then the eigenvalues $\lambda_i(\mathsf A^\dagger \mathsf A)$ are non-negative.
Furthermore, the singular values $\sigma_i(\mathsf A)=\sqrt{\lambda_i(\mathsf A^\dagger \mathsf A)}$.
\end{fact}

\begin{fact}
\label{fact:bipartite_svalues}
Let $V$,$W$ be two inner product spaces and $\mathsf A:V\to W$ be a linear operator and let $\mathsf B$ be defined by the expression,
\begin{align*}
\mathsf B=
\begin{bmatrix}
0 && \mathsf A\\
\mathsf{A}^{\dagger} && 0
\end{bmatrix}
\end{align*}
then for any $i\in \set{1, \dots, r}$, $\sigma_i(\mathsf A)=\lambda_i(\mathsf B)$ where $r=\rank(\mathsf A)$.
\end{fact}

\begin{proof}
For completeness, we refer \prettyref{app:proof_of_bipartite_svalues} for a proof of this.
\end{proof}

\subsubsection{Simplicial Complexes}

\begin{definition}
A simplicial complex $X$ is  a set system which consists of a ground set $V$ and a downward closed collection of subsets of $V$, i.e., if $s \in X$ and $t \subseteq s$ then $t \in X$.
The sets in $X$ are called the faces of $X$.

We define a level/slice $X(l)$ of the simplicial complex $X$ as  $X(l)=\set{s\in X||s|=l}$.
Note that for the simplicial complex corresponding to the hypergraph, the top level $X(k)$ is the set of $k$-uniform hyperedges and the ground set of vertices\footnote{We shall often simply write $v$ for a face $\set{v}\in X(1)$} is denoted $X(1)$.
By convention we have that $X(0)=\set{\emptyset}$.
Similarly, we define $X(\leq l) = \set{s\in X||s|\leq l}$.

We call a simplicial complex $X$ as $k$-dimensional if $k$ is the smallest integer for which $X(\leq k)=X$.\footnote{We shall often write $X(\leq k)$ for $X$ to stress the fact that X is $k$-dimensional}
A $k$-dimensional simplicial complex $X$ is a pure simplicial complex if for all $s \in X$ there exists $t \in X(k)$ such that $s \subseteq t$.
\end{definition}

Given a $k$-uniform hypergraph $H=(V,E)$, we obtain a pure simplicial complex $X$ where the ground set is $V$ and we downward close the set system $E$ of hyperedges. 
Given a distribution $\Pi_k$ on the hyperedges, we have an induced distribution $\Pi_l$ on sets $s$ in level $X(l)$ given by,
\begin{align*}
    \Pi_l(s)=\frac{1}{\binom{k}{l}}\sum_{e\in E|s\subseteq e}\Pi_k(e).
\end{align*}
We refer to the joint distribution as $\Pi = \paren{\Pi_k,\Pi_{k-1},\hdots,\Pi_1}$.
If the input hypergraph is unweighted then we take the distribution $\Pi_k$ to be the uniform distribution on $X(k)$.
We thus obtain a weighted simplicial complex $(X,\Pi)$.
We refer to $(X,\Pi)$ as the (weighted\footnote{Whenever it is clear from the context, we use $X$ in place of $(X,\Pi)$ for the sake of brevity.}) simplicial complex \emph{induced} by $(H,\Pi_k)$.

\begin{lemma}[Folkore]
\label{lem:pi_level_relation}
For any non-negative integer $m\leq l$ and any $s\in X(m)$, $\sum_{t\in X(l)|t\supseteq s}\Pi_l(t)=\binom{l}{m}\Pi_m(s)$.
\end{lemma}

\begin{proof}
This essentially follows from the definitions.
For completeness, we refer \prettyref{app:proof_of_pi_level_relation} for a proof of this.
\end{proof}

In this work, we consider a notion of expansion for weighed simplicial complexes called link expansion.
To that end we first define the notion of link of a complex and its skeleton.

\begin{definition}
For a simplicial complex $X$ and some $s\in X$, $X_s$ denotes the link complex of $s$ defined by $X_s=\set{t\setminus s|s\subseteq t\in X}$.
The skeleton of a link $X_s$ for a face $s\in X(\leq k-2)$ (where $k$ is the length of the largest face) denoted by $G(X_s)$ is a weighted graph with vertex set $X_s(1)$, edge set $X_s(2)$ and weights proportional to $\Pi_2$.
\end{definition}

\begin{definition}[$\gamma$-HDX, \cite{DK17}]
\label{def:local_hdx}
A simplicial complex $X(\leq k)$ is a $\gamma$-High Dimensional Expander ($\gamma$-HDX) if for all $s\in X(\leq k-2)$, the second singular value of the adjacency matrix of the graph $G(X_s)$ (denoted $\sigma_2(G(X_s))$) satisfies $\sigma_2(G(X_s))\leq \gamma$.
We refer to $1-\gamma$ as the link-expansion of $X$. 
\end{definition}

\begin{definition}[Weighted inner product]
Given two functions $f,g \in \mathbb{R}^{S}$, i.e., $f,g:S\to \R$ and a measure $\mu$ on $S$, we define the weighted inner product of these functions as,
\begin{align*}
    \inprod{f,g}_{\mu} = \E_{s \sim \mu}[f(s)g(s)] = \sum_{s \in S}f(s)g(s)\mu(s) \mper
\end{align*}
We drop the subscript $\mu$ from $\inprod{\cdot,\cdot}_\mu$ whenever $\mu$ is clear from context.
\end{definition}

\paragraph{A note on convention.}
In this paper we will use the weighted inner product between two functions $f, g\in \R^{X(m)}$ on levels $X(m)$ of the simplicial complex $X$ under consideration and with the measure $\Pi_m$, unless otherwise specified.
In particular, for any linear operator $\mathsf A:\R^{X(m)}\to \R^{X(l)}$ the adjoint $\mathsf A^\dagger$ and the $i$-th largest singular value $\sigma_i(\mathsf A)$ are with respect to this inner-product.

This choice of inner product is particularly convenient for dealing with random-walks on a simplicial complex, for example, see \prettyref{lem:updown_adj} which shows that the ``down-walk'' is the adjoint of the ``up-walk''.
Roughly speaking, this is the result of weighing each component using the appropriate ``degree''.

\subsubsection{Walks on a Simplicial Complex}

\begin{definition}[Up and Down operators]
Given a simplicial complex $(X,\Pi)$, we define the up operator $\mathsf{U}_i: \mathbb{R}^{X(i)} \rightarrow \mathbb{R}^{X(i+1)}$ that acts on a function $f \in \mathbb{R}^{X(i)}$ as,
\begin{align*}
    [\mathsf{U}_i f](s) = \E_{s' \in X(i), s'\subseteq s}[f(s')] = \frac{1}{i+1}\sum_{x \in s}f(s \setminus \set{x})
\end{align*}
and the down operator $\mathsf{D}_{i+1}: \mathbb{R}^{X(i+1)} \rightarrow \mathbb{R}^{X(i)}$ that acts on a function $g \in \mathbb{R}^{X(i+1)}$ as,
\begin{align*}
    [\mathsf{D}_{i+1}g](s) = \E_{s' \sim \Pi_{i+1}, s'\supset s}[g(s')] = \frac{1}{i+1}\sum_{x \notin s}g(s \cup \set{x})\frac{\Pi_{i+1}(s \cup \set{x})}{\Pi_i(s)} \mper
\end{align*}
\end{definition}

As a consequence of the definition of the up and down operators the following statement holds.
\begin{lemma}[Folklore]
\label{lem:updown_adj}
$\mathsf U_{i}^\dagger =\mathsf D_{i+1}$.
\end{lemma}
\begin{proof}
For completeness, we reproduce a proof of this in \prettyref{app:proof_of_updown_adj}.
\end{proof}

The up operator $\mathsf U_i$ can be thought of as defining a random-walk moving from $X(i+1)$ to $X(i)$ where a subset of size $i$ is selected uniformly for a given face $s\in X(i+1)$.
Similarly, the down operator $\mathsf D_{i+1}$ can be thought of as defining a random walk moving from $X(i)$ to $X(i+1)$ where a superset $s'\in X(i+1)$ of size $i+1$ is selected for a given face $s\in X(i)$ with probability $\frac{\Pi_{i+1}(s')}{\Pi_i(s)}$.
This leads us to the following definition.

\begin{definition}
\label{def:updown_graph}
Given a simplicial complex $(X,\Pi)$ and its two levels  $X(m)$, $X(l)$, we define a bipartite graph on $X(m)\cup X(l)$ as $B_{m,l}=(X(m) \cup X(l),E_{m,l},w_{m,l})$ where $E_{m,l} = \set{\set{s,t} | s \in X(m),t \in X(l), \text{ and } s \subseteq t}$ and $m \leq l$.
The weight of an edge $\set{s,t}$ where $s\in X(m)$ and $t\in X(l)$ is given by $w_{m,l}(s,t)= \binom{k}{l}\Pi_{l}(t)$.
\end{definition}
As we will show in \prettyref{fact:b_equals_n}, in the random-walk on $B_{m,l}$ the block corresponding to the transition from a vertex in $X(m)$ to a vertex in $X(l)$ is the up walk (i.e., the down operator) and the block corresponding to the transition from a vertex in $X(l)$ to a vertex in $X(m)$ is the down-walk (i.e., the up operator).

Now, we define the $B^{(2)}_{m,l}$ graph such that the random-walk on it corresponds to the two step walk starting from vertices in $X(m)$ on $B_{m,l}$, i.e., the random-walk on $B^{(2)}_{m,l}$ corresponds to a up-walk followed by a down-walk.
\prettyref{fact:b2_equals_updown} shows that this correspondence indeed holds.

\begin{definition}
\label{def:order2_updown}
Given a simplicial complex $(X,\Pi)$ and its two levels  $X(m)$, $X(l)$ with $m \leq l$, we define a graph on $X(m)$ as $B^{(2)}_{m,l}=(X(m), E^{(2)}_{m,l},w^{(2)}_{m,l})$ where 
\begin{align*}
    E^{(2)}_{m,l} = \set{\set{s,t} | s, t \in X(m) \text{ and } \exists s'\in X(l) \text{ such that }s'\supseteq s\cup t}.
\end{align*}
The weight of an edge $\set{s,t}$ where $s, t\in X(m)$ is given by 
\begin{align*}
w^{(2)}_{m,l}(s,t)= \sum_{s'\supseteq s\cup t} w_{m,l}(s,s')=\binom{k}{l}\sum_{s'\supseteq s\cup t} \Pi_{l}(s').
\end{align*}

The normalized adjacency matrix corresponding to $B^{(2)}_{m,l}$ is denoted by $\mathsf A^{(2)}_{m,l}$.
\end{definition}

\begin{definition}[Up-Down Walk, \cite{DK17}]
\label{def:canonical_walk}
For positive integers $m\leq l$, let $\mathsf{D}_{m,l}$ denote the product
\begin{align*}
    \mathsf D_{m+1}\mathsf D_{m+2}\hdots \mathsf D_{l-1}\mathsf D_l
\end{align*}
and let $\mathsf{U}_{l,m}$ denote the product
\begin{align*}
    \mathsf U_{l-1}\mathsf U_{l-2}\hdots \mathsf U_{m+1}\mathsf U_{m}.
\end{align*}
We denote the following walk between $X(m)$ and $X(l)$ as $\mathsf N_{m,l}$,
\begin{align*}
\mathsf{N}_{m,l}=
\begin{bmatrix}
0 && \mathsf{D}_{m,l}\\
\mathsf{U}_{l,m} && 0
\end{bmatrix}=
\begin{bmatrix}
0 && \mathsf{D}_{m,l}\\
\mathsf{D}_{m,l}^{\dagger} && 0
\end{bmatrix},
\end{align*}
where the second equality is a consequence of \prettyref{lem:updown_adj}.
The up-down walk on $X(m)$ through $X(l)$ is a random-walk on $X(m)$ whose transition matrix (denoted $\mathsf N^{(2)}_{m,l}$) is given by $\mathsf N^{(2)}_{m,l}=\mathsf D_{m,l}\mathsf U_{l,m}$.
\end{definition}

\begin{fact}
\label{fact:b_equals_n}
The transition matrix for random-walk on $B_{m,l}$ is $\mathsf N_{m,l}$.
\end{fact}
\begin{proof}
For completeness, we refer \prettyref{app:proof_of_b_equals_n} for a proof of this.
\end{proof}

\begin{fact}
\label{fact:b2_equals_updown}
The transition matrix for random-walk on $B^{(2)}_{m,l}$ is $\mathsf {D}_{m,l}\mathsf U_{l,m}$.
\end{fact}
\begin{proof}
For completeness, we refer \prettyref{app:proof_of_b2_equals_updown} for a proof of this.
\end{proof}

%% file: cheeger.tex
\section{Computing Sparse Cut in Hypergraphs}
\label{sec:sparse_cut_in_non-hdx}
\prettyref{thm:sparse_cut} shows an analogue of Cheeger's inequality based on the eigenvalues of up-down  walks $\mathsf{N}_{1,l}$.

\begin{theorem}
\label{thm:sparse_cut}
Let $H=(V,E)$ be a $k$-uniform hypergraph such that the induced simplicial complex $X$ has a up-down  walk $\mathsf{N}^{(2)}_{1,l}$ such that $\lambda_2(\mathsf{N}_{1,l})=1-\varepsilon$ for some $\varepsilon >0$ and some $l \in \set{2,3,\hdots,k}$.
Then $\frac{\varepsilon}{k}\leq \phi_H\leq 4\sqrt{\varepsilon}$.
Furthermore there is an algorithm which on input $H$, outputs a set $S \subset V$ such that $\phi_{H}(S) \leq 4\sqrt{\varepsilon}$ in $\mathsf{poly}(|V|,|E|)$ time where $\mathsf{poly}$ is a polynomial.
\end{theorem}
Towards proving this theorem, we start with a simple but useful statement, i.e., \prettyref{fact:svalue_ineq} that allows us to work with the $\mathsf{D}_{1,2}$ walk instead of the $\mathsf{N}_{1,l}$ walk for some $l \in \set{3,4,\hdots,k}$.

\begin{fact}[Folklore]
\label{fact:svalue_ineq}
Let $\mathsf A\in \R^{n\times m}$, $\mathsf B\in \R^{m\times p}$ and $\sigma_i$ denote the $i^{th}$ singular value then we have that,
\begin{align*}
    \sigma_i(\mathsf{AB})\leq \sigma_1(\mathsf A)\sigma_i(\mathsf B),
\end{align*}
and
\begin{align*}
    \sigma_i(\mathsf {AB})\leq \sigma_i(\mathsf A)\sigma_1(\mathsf B),
\end{align*}
for $i=1,\dots,r$, where $r=\rank(\mathsf {AB})$.
\end{fact}

\begin{proof}
For completeness, we refer \prettyref{app:proof_of_svd_product} for a proof of this.
\end{proof}

\begin{corollary}
\label{cor:one_to_two}
If $\sigma_2\paren{\mathsf{D}_{1,l}} = 1-\varepsilon$ for an arbitrary $l \in \set{2,3,\hdots,k}$, we have that $\sigma_2(\mathsf{D}_{1,2}) \geq 1-\varepsilon$.
\end{corollary}

\begin{proof}
The proof follows by using \prettyref{fact:svalue_ineq} and writing $\mathsf{D}_{1,l} = \mathsf{D}_{1,2}\mathsf{D}_{2,l}$ as,
\begin{align*}
    \sigma_2\paren{\mathsf{D}_{1,l}} = \sigma_2\paren{\mathsf{D}_{1,2}\mathsf{D}_{2,l}} \overset{\prettyref{fact:svalue_ineq}}{\leq} \sigma_2(\mathsf{D}_{1,2})\sigma_1(\mathsf{D}_{2,l}) = \sigma_2\paren{\mathsf{D}_{1,2}}
\end{align*}
where the last equality holds since $\sigma_1\paren{\mathsf{D}_{2,l}}=1$.
\end{proof}

Next, we show that we can use this information about $\sigma_2(\mathsf{D}_{1,2})$ to compute a set $S \subset V$ such that its expansion in the $B^{(2)}_{1,2}$ graph is at most $2\sqrt{\varepsilon}$.

\begin{lemma}
\label{lem:cheeger_cut}
If $\sigma_2(\mathsf{D}_{1,2}) = 1-\varepsilon$ for some $\varepsilon\in (0,1)$, then there exists a set $S \subseteq X(1)$ such that $\phi_{B^{(2)}_{1,2}}(S) \leq 2\sqrt{\varepsilon}$.
Furthermore, there is a $\mathsf{poly}(|V_{B^{(2)}_{1,2}}|,|E_{B^{(2)}_{1,2}}|)$ time algorithm to compute such a set $S$.
\end{lemma}

\begin{proof}
Let $u,v$ be the left and the right singular vectors of $\mathsf{D}_{1,2}$ with singular value $1-\varepsilon$, i.e., $u^\dagger \mathsf{D}_{1,2}=(1-\varepsilon)v^\dagger$ and $\mathsf{D}_{1,2}v=(1-\varepsilon)u$.
Let $\mathsf{W}$ be the random-walk matrix for $B^{(2)}_{1,2}$.
First observe that $\mathsf{W=D_{1,2}U_{2,1}=D_{1,2}D^\dagger_{1,2}}$ where the first equality is a consequence of \prettyref{fact:b2_equals_updown} and the second equality is a consequence of \prettyref{lem:updown_adj}.
Hence, we have 
\begin{align*}
    u^\dagger \mathsf{W}=u^\dagger \mathsf{D_{1,2}D^\dagger_{1,2}}=(1-\varepsilon)v^\dagger \mathsf{D^\dagger_{1,2}}=(1-\varepsilon)^2u^\dagger,
\end{align*}
i.e., $u$ is a left eigenvector of $\mathsf W$ with eigenvalue $(1-\varepsilon)^2$.
Since, the matrices $\mathsf{W}$ and $\mathsf{A^{\paren{2}}_{1,2}}$ are similar, $\mathsf{A^{\paren{2}}_{1,2}}$ has a eigenvalue $(1-\varepsilon)^2$.
Let $x$ be the eigenvector of $\mathsf{A^{\paren{2}}_{1,2}}$ with eigenvalue $(1-\varepsilon)^2$.
Note that $x$ and $\mathbb 1$ are eigenvectors of $\mathsf{A^{\paren{2}}_{1,2}}$ with distinct eigenvalues and the matrix $\mathsf{A^{\paren{2}}_{1,2}}$ is also symmetric. Therefore, $x$ and $\mathbb 1$ are
orthogonal with respect to the inner-product given by the uniform distribution on the vertices of $B^{(2)}_{1,2}$.
Now, we use Fiedler's algorithm with the vector $x$ to obtain a set $S\subseteq V$ such that
\begin{align*}
    \phi_{B^{(2)}_{1,2}}(S) \leq \sqrt{2(1-(1-\varepsilon)^2)} \leq 2\sqrt{\varepsilon}.
\end{align*}
\end{proof}

A natural choice for our set $S$ with low conductance in input hypergraph is this set $S$ guaranteed by the Fiedler's algorithm for which $\phi_{B^{(2)}_{1,2}}(S)$ is small. 
We show in \prettyref{lem:uniform_count2} that $B^{(2)}_{1,2}$ is a weighted graph where the weight of an edge between two distinct vertices in $X(1)$ is the multiplicity of that edge in the construction of $B^{(2)}_{1,2}$ graph.
We note that a hyperedge $e$, induces a clique on the vertices in the hyperedge $e$, in the $B^{(2)}_{1,2}$ graph.
This is commonly known as the clique expansion of the hypergraph.

\begin{lemma}
\label{lem:uniform_count2}
For any $k$-uniform hypergraph $H=(V,E)$, let $X$ be the induced simplicial complex and let $\set{s,t}$ be an edge in $B^{(2)}_{m,l}$ with $s,t\in X(m)$.
Then $w(s,t)=\binom{k-\abs{s\cup t}}{l-\abs{s\cup t}}\sum_{e\in E|s\cup t\subseteq e}\Pi_k(e)$ and 
$\mathsf {deg}_{B^{(2)}_{m,l}}(s)={\binom{l}{m}}^2\frac{\binom{k}{l}}{\binom{k}{m}}\sum_{e\in E|e\supseteq s}\Pi_k(e)$.
\end{lemma}
\begin{proof}
For completeness, we refer \prettyref{app:proof_of_uniform_count2} for a proof of this.
\end{proof}

Now in \prettyref{lem:related_boundary_edges} we show how the weight of edges cut in the boundary of the weighted graph $B^{(2)}_{1,2}$ and the input hypergraph are related.

\begin{lemma}
\label{lem:related_boundary_edges}
Given a set $S \subset X(1)$ we have that
\begin{align*}
    (k-1)\Pi_k\paren{\partial_H(S)} \leq w(\partial_{B^{(2)}_{1,2}}(S))\mper
\end{align*}
\end{lemma}

\begin{proof}
By \prettyref{lem:uniform_count2}, $B^{(2)}_{1,2}$ is a weighted graph where the weight $w(i,j)$ of an edge $\set{i,j}$ where $i\ne j$ is given by $w(i,j)=\sum_{e\in E|\set{i,j}\subseteq e}\Pi_k(e)$.
Therefore, to compute $w(\partial_{B^{(2)}_{1,2}}(S))$ we sum over all $i\in S$ and $j\in V\setminus S$, the number of hyperedges containing $\set{i,j}$, i.e.,
\begin{align*}
    w(\partial_{B^{(2)}_{1,2}}(S))=\sum_{i \in S, j\in V\setminus S}\sum_{\substack{e\in H \\e\supseteq \set{i,j}}}\Pi_k(e) = \sum_{e\in H}\sum_{\substack{i\in S,j\in V\setminus S\\ \set{i,j}\subseteq e}}\Pi_k(e) && \paren{\text{Exchanging order of summation}}.
\end{align*}
Now, we note that the number of $\set{i,j}\subseteq e$ where $i\in S$ and $j\in V\setminus S$ is non-zero if and only if $e\in \partial_{H}(S)$, and hence,
\begin{align}
\label{eq:partialH}
    w(\partial_{B^{(2)}_{1,2}}(S))=\sum_{e\in \partial_{H}(S)}\sum_{\substack{i\in S,j\in V\setminus S\\ \set{i,j}\subseteq e}}\Pi_k(e).
\end{align}

Now, let $e \cap S = \set{i_1,i_2,\hdots,i_t}$ for some $t \in \set{1,2,\hdots,k-1}$.
For the lower bound, we note that the number of $\set{i,j}\subseteq e$ where $i\in S$ and $j\in V\setminus S$ is $t(k-t)$.
Therefore, for some $e\in \partial_{H}(S)$ we have the minimum value of $t(k-t)$ as $k-1$ and hence using eqn. \prettyref{eq:partialH} we get,
\begin{align*}
    w(\partial_{B^{(2)}_{1,2}}(S))\geq \sum_{e\in \partial_{H}(S)} (k-1)\Pi_k(e)=(k-1)\Pi_k\paren{\partial_{H}(S)} &\paren{\Pi_k(\partial_H(S))=\sum_{e\in \partial_{H}(S)}\Pi_k(e)}.
\end{align*}
\end{proof}
We now show an upper bound for boundary of $B^{(2)}_{1,l}$ in terms of the boundary of $H$.

\begin{lemma}
\label{lem:boundary_upper}
For any $l$, such that $2\leq l\leq k$, Given a set $S \subset X(1)$ we have that
\begin{align*}
    w(\partial_{B^{(2)}_{1,l}}(S)) \leq \binom{k}{l}\binom{l}{2}\Pi_k\paren{\partial_H(S)} \mper
\end{align*}
\end{lemma}
\begin{proof}
By $\prettyref{lem:uniform_count2}$, $B^{(2)}_{1,l}$ is a weighted graph where the weight $w(i,j)$ of an edge $\set{i,j}$ where $i\ne j$ is given by $w(i,j)=\binom{k-2}{l-2}\sum_{e\in E|\set{i,j}\subseteq e}\Pi_k(e)$.
Therefore, to compute $w(\partial_{B^{(2)}_{1,l}}(S))$ we sum over all $i\in S$ and $j\in V\setminus S$, the number of hyperedges containing $\set{i,j}$, i.e.,
\begin{align*}
    w(\partial_{B^{(2)}_{1,l}}(S))=\binom{k-2}{l-2}\sum_{i \in S, j\in V\setminus S}\sum_{\substack{e\in H \\e\supseteq \set{i,j}}}\Pi_k(e) = \binom{k-2}{l-2}\sum_{e\in H}\sum_{\substack{i\in S,j\in V\setminus S\\ \set{i,j}\subseteq e}}\Pi_k(e) && \paren{\text{Rearranging}}.
\end{align*}
Again, we note that the number of $\set{i,j}\subseteq e$ where $i\in S$ and $j\in V\setminus S$ is non-zero if and only if $e\in \partial_{H}(S)$, and hence,
\begin{align}
\label{eq:general_partialH}
    w(\partial_{B^{(2)}_{1,2}}(S))=\binom{k-2}{l-2}\sum_{e\in \partial_{H}(S)}\sum_{\substack{i\in S,j\in V\setminus S\\ \set{i,j}\subseteq e}}\Pi_k(e).
\end{align}
Now, we see that the number of subsets $\set{i,j}\subseteq e$ where $i\in S$ and $j\in V\setminus S$ is upper bounded by the number of cardinality two subsets of $e$, i.e., $\binom{k}{2}$ subsets and using eqn. \prettyref{eq:general_partialH} we get,
\begin{align*}
    w(\partial_{B^{(2)}_{1,2}}(S))\leq& \binom{k-2}{l-2} \sum_{e\in \partial_{H}(S)} \sum_{\set{i,j}\subseteq e}\Pi_k(e)\\
    &= \binom{k-2}{l-2} \sum_{e\in \partial_{H}(S)} \binom{k}{2}\Pi_k(e)\\
    &=\binom{k}{2}\binom{k-2}{l-2}\Pi_k\paren{\partial_{H}(S)}&\paren{\Pi_k(\partial_H(S))=\sum_{e\in \partial_{H}(S)}\Pi_k(e)}\\
    &=\binom{k}{l}\binom{l}{2}\Pi_k\paren{\partial_{H}(S)}&\paren{\binom{k-2}{l-2}\binom{k}{2}=\binom{k}{l}\binom{l}{2}}.
\end{align*}
\end{proof}

Next, in \prettyref{lem:relate_expansion} we will use these bounds to analyze the expansion of this set $S$ in the input hypergraph.

\begin{lemma}
\label{lem:relate_expansion}
For an arbitrary set $S \subset X(1)$, we have that $\phi_H(S) \leq 2\phi_{B^{(2)}_{1,2}}(S)$.
\end{lemma}

\begin{proof}
We start by comparing the numerator in the expressions for expansion of the given arbitrary set $S$ in original hypergraph  $\abs{\partial_H(S)}$ and in the  $B_{1,2}^{(2)}$ graph, i.e., $w(\partial_{B^{(2)}_{1,2}}(S))$.
Using \prettyref{lem:related_boundary_edges} we have that, $\Pi_k\paren{\partial_H(S)}  \leq \frac{1}{(k-1)}\cdot w(\partial_{B_{1,2}^{(2)}}(S))$. 

Next, we compare the denominators in the respective expression for expansions, i.e., $\mathsf{vol}_H(S)$ and $\mathsf{vol}_{B^{(2)}_{1,2}}(S)$. 
For the hypergraph, by definition we have that $\mathsf{vol}_H(S) = \sum_{i \in S}\mathsf{deg}(i)$.
By \prettyref{lem:uniform_count2} we have
\begin{align*}
\mathsf {vol}_{B^{(2)}_{1,2}}(S)=\sum_{i\in S}\mathsf{deg}_{B^{(2)}_{1,2}}(i)=\sum_{i\in S}\binom{2}{1}^2\frac{k(k-1)}{2k}\mathsf{deg}_{H}(i)=2(k-1)\mathsf {vol}_H(S).
\end{align*}

Now, putting everything together we have that,
\begin{align*}
    \phi_H(S) = \frac{\Pi_k\paren{\partial_H(S)}}{\mathsf{vol}_H(S)} =  2(k-1)\frac{\Pi_k\paren{\partial_H(S)}}{\mathsf{vol}_{B_{1,2}^{(2)}}(S)} \leq 2\cdot \frac{(k-1)}{(k-1)}\cdot\frac{w(\partial_{B_{1,2}^{(2)}}(S))}{\mathsf{vol}_{B_{1,2}^{(2)}}(S)} = 2\phi_{B_{1,2}^{(2)}}(S) \mper
\end{align*}
\end{proof}

\begin{lemma}
\label{lem:lowerbound_expansion}
For an arbitrary set $S \subset X(1)$, we have that $\phi_H(S) \geq \frac{2}{k}\phi_{B^{(2)}_{1,l}}(S)$.
\end{lemma}
\begin{proof}
Using \prettyref{lem:boundary_upper} we have that, $\Pi_k\paren{\partial_H(S)}  \geq \frac{1}{\binom{k}{l}\binom{l}{2}}\cdot w(\partial_{B_{1,l}^{(2)}}(S))$.
For the hypergraph, by definition we have that $\mathsf{vol}_H(S) = \sum_{i \in S}\mathsf{deg}(i)$.
And by \prettyref{lem:uniform_count2} we have
\begin{align*}
\mathsf {vol}_{B^{(2)}_{1,l}}(S)=\sum_{i\in S}\mathsf{deg}_{B^{(2)}_{1,l}}(i)=\sum_{i\in S}\binom{l}{1}^2\binom{k}{l}\frac{1}{k}\mathsf{deg}_{H}(i)=\binom{k}{l}\frac{l^2}{k}\mathsf {vol}_H(S).
\end{align*}
Now, putting everything together we have that,
\begin{align*}
    \phi_H(S) = \frac{\Pi_k\paren{\partial_H(S)}}{\mathsf{vol}_H(S)} = \binom{k}{l}\frac{l^2}{k}\frac{\Pi_k\paren{\partial_H(S)}}{\mathsf{vol}_{B_{1,l}^{(2)}}(S)} \geq \frac{l^2\binom{k}{l}}{\binom{k}{l}\binom{l}{2}k}\cdot\frac{w(\partial_{B_{1,l}^{(2)}}(S))}{\mathsf{vol}_{B_{1,l}^{(2)}}(S)} = \frac{2}{k}\phi_{B_{1,l}^{(2)}}(S) \mper
\end{align*}
\end{proof}

\begin{proof}[Proof of \prettyref{thm:sparse_cut}]
\label{proof:sparse_cut_theorem}
First, we note by \prettyref{fact:svalue_eigenvalue}, $1-\varepsilon\leq \sqrt{1-\varepsilon}\leq\sqrt{\lambda_2(\mathsf N^{(2)}_{1,l})}=\sigma_2(\mathsf D_{1,l})$.

Now, using \prettyref{cor:one_to_two} we conclude that $\sigma_2(\mathsf{D}_{1,2})=1-\varepsilon'\geq1-\varepsilon$ for some $\varepsilon'\leq \varepsilon$.
Further, in \prettyref{lem:cheeger_cut}, we show that we can use this information about the spectrum of $\mathsf D_{1,2}$ to compute a set $S \subset V$ such that its expansion the $B_{1,2}^{(2)}$ graph is at most $2\sqrt{\varepsilon}$.
We fix this to be the set $S$ we return in our sparse cut.
In \prettyref{lem:relate_expansion} we show that expansion of this set $S$ in the input hypergraph is at most $2\phi_{B_{1,2}^{(2)}}(S)$ and hence 
\begin{align*}
    \phi_H(S) \leq 2\phi_{B_{1,2}^{(2)}}(S) \leq 4\sqrt{\varepsilon}  \mper
\end{align*}

Now, by \prettyref{fact:b2_equals_updown} the matrices $\mathsf N^{(2)}_{1,l}$ and $\mathsf A^{(2)}_{1,l}$ are similar and hence have the same eigenvalues and therefore by Cheeger's inequality, we have $\phi_{B^{(2)}_{1,l}}\geq \frac{\varepsilon}{2}$.
Therefore by \prettyref{lem:lowerbound_expansion}, we have 
\[
\phi_H\geq \frac{2}{k}\phi_{B^{(2)}_{1,l}}\geq \frac{\varepsilon}{k}.
\]
\end{proof}

%% file: counter-example.tex
\section{An expanding hypergraph with walks having small spectral gap}
\label{sec:example_hdx_non-swap_walks}
\subsection{Additional Preliminaries}

In this section, we consider a 'non-lazy' version of the up-down walk.
While typically for a walk on the graph to be non-lazy we require that there be no transition from a vertex to itself, we obtain the swap-walks by imposing an even stronger condition where we don't allow any face to have a transition to another face with a non-empty intersection with the starting face. 

\begin{definition}[Swap-walk, \cite{AJT19,DD19}]
Given a $k$-dimensional simplicial complex $(X,\Pi)$, for non-negative integers $m,l$ such that $l+m\leq k$ we define the swap-walk denoted $\mathsf{S}_{m,l}:\R^{X(l)}\to \R^{X(m)}$ that acts on a $f\in \R^{X(l)}$ as,
\begin{align*}
    [S_{m,l}f](s)=\E_{s'\sim \Pi_{m+l}|s'\supseteq s}f(s'\setminus s).
\end{align*}
\end{definition}

\begin{lemma}[\cite{AJT19}]
\label{lem:swap_adj}
$\mathsf S_{m,l}^\dagger =\mathsf S_{l,m}$.
\end{lemma}
\begin{proof}
For completeness, we reproduce a proof of this in \prettyref{app:proof_of_swap_adj}.
\end{proof}

Again, the swap-walk $\mathsf S_{m,l}$ can be thought of as defining a random-walk moving from $X(m)$ to $X(l)$ where we first move from $s\in X(m)$ to a superset $s''\in X(m+l)$ with probability $\frac{\Pi_{m+l}(s'')}{\Pi_m(s)}$ and then determistically move to $s'=s''\setminus s$, i.e., we move from face $s\in X(m)$ to a disjoint face $s'\in X(l)$ with probability $\frac{\Pi_{m+l}(s\sqcup s')}{\Pi_{m}(s)}$. 
This leads us to the following definition for swap-graphs.

\begin{definition}[Swap-graph, Section 6 in \cite{AJT19}]
\label{def:swap_graph}
Given a simplicial complex $(X,\Pi)$ and its two levels  $X(m),X(l)$, the swap-graph (denoted $G_{m,l}$) is defined as a bipartite graph $G_{m,l}=(X(m) \cup X(l),E(m,l),\mathrm w_{m,l})$ where $E(m,l) = \set{\set{s,t} | s \in X(m),t \in X(l), \text{ and } s \sqcup t \in X(m+l)}$.
The weight function is defined as, $\mathrm w_{m,l}(s,t)= \frac{\Pi_{m+l}(s \sqcup t)}{\binom{m+l}{m}}$.
\end{definition}

The random-walk matrix corresponding to these walks denoted $\mathsf W_{m,l}$ is a matrix of size $(|X(m)|+|X(l)|)\times (|X(m)|+|X(l)|)$ 
and is given by, 
\begin{equation}
\label{eq:walkmatrix_swap}
\mathsf W_{m,l}=
\begin{bmatrix}
0 && \mathsf{S}_{m,l}\\
\mathsf{S}_{l,m} && 0
\end{bmatrix}=
\begin{bmatrix}
0 && \mathsf{S}_{m,l}\\
\mathsf{S}_{m,l}^{\dagger} && 0
\end{bmatrix},
\end{equation}
where the last equality is a consequence of \prettyref{lem:swap_adj}.

Arora, Barak, and Steurer \cite{ABS10} introduced the notion of threshold-rank of a graph.
\begin{definition}[Threshold rank of a graph, \cite{ABS10} ]
\label{def:threshold_rank}
Given a weighted graph $G=(V,E,w)$ and its normalized random-walk matrix $\mathsf W$ such that $\lambda_n(\mathsf W) \leq \lambda_{n-1}(\mathsf W) \leq \hdots \leq \lambda_1(\mathsf W)=1$ and a threshold $\tau \in (0,1]$, we define the $\tau$-threshold rank of the graph $G$ (denoted by $\rank_{\geq \tau}(\mathsf W)$) as $\rank_{\geq \tau} (\mathsf W)=\abs{\set{i|\lambda_i(\mathsf W) \geq \tau}}$.
\end{definition}

To extend the techniques from \cite{BRS11,GS11} to $k$-CSPs, \cite{AJT19} proposed that specific sets of swap-walks be used to define an analogue of threshold-rank for hypergraphs called $(\tau,r)$-splittability.
To define the set of swap-walks needed they consider the following class of binary trees.
\begin{definition}[$k$-splitting tree, Section 7 in \cite{AJT19}]
\label{def:splitting_tree}
A binary tree $\mathcal{T}$ given with its labeling is called a $k$-splitting tree if
\begin{itemize}
    \item $\mathcal{T}$ has exactly $k$ leaves.
    \item The root of $\mathcal{T}$ is labeled with $k$ and all other vertices in $\mathcal{T}$ are labeled with a positive integer.
    \item All the leaves are labeled with $1$.
    \item The label of every internal node of $\mathcal{T}$ is the sum of the labels of its two children.
\end{itemize}
\end{definition}
Now, given a $k$-splitting tree $\mathcal T$ define a set of swap-walks depending on $\mathcal T$ and threshold-rank for that set as follows.
\begin{definition}[Swap-graphs in a tree, Section 7 in \cite{AJT19}]
\label{def:swap_graph_in_tree}
For a simplicial complex $X(\leq k)$ and a $k$-splitting tree $\mathcal{T}$, we consider all swap-graphs (denoted $\mathsf{Swap}(\mathcal{T},X)$) from $X(a)$ to $X(b)$ where $a$ and $b$ are labels of a non-leaf node in $\mathcal{T}$.
Further, we extend the definition of threshold rank as 
\begin{align*}
    \rank_{\geq \tau}\paren{\mathsf{Swap}(\mathcal{T},X)} = \max_{G \in \mathsf{Swap}(\mathcal{T},X)}\rank_{\geq \tau}(G)\mper
\end{align*}
\end{definition}
Finally, define $(\tau, r)$-splittability by considering all such sets of swap-walks.
\begin{definition}[$(\tau,r)$-splittability, Definition 7.2 in \cite{AJT19}] 
\label{def:splittability}
A $k$-uniform hypergraph with an induced simplicial complex $X(\leq k)$ is said to be $(\tau,r)$-splittable if there exists some $k$-splittable tree $\mathcal{T}$ such that $\rank_{\geq \tau}\paren{\mathsf{Swap}(\mathcal{T},X)} \leq r$.
\end{definition}

\subsection{The main results}
In \prettyref{thm:non_splittable}, we show an example of an expanding hypergraph such that for all $m,l$ such that $m+l\leq k$ the swap-walk from $X(m)$ to $X(l)$ in the corresponding simplicial complex has its top $r$ singular value as $1$ (for $r \approx n/k$) if either $m,l\geq 2$ or $m=k-l$.

\begin{theorem}
\label{thm:non_splittable}
For any positive integers $r,k$ with $r\geq 2,k\geq 3$, there exists an $k$-uniform hypergraph $H$ on $n(=r(k-1)+1)$ vertices such that $\phi_H\geq \frac{1}{k}$ and
for any $m, l$ such that $m+l\leq k$, if either $m,l\geq 2$ or $m=k-l$ then $\lambda_r(G_{m,l})=\sigma_r(\mathsf S_{m,l})=1$, where $\mathsf S_{m,l}$,$G_{m,l}$ are the swap-walk and the swap-graph on the induced simplicial complex $X$, respectively.
\end{theorem}
Now, \prettyref{cor:non_splittable} is a simple consequence of \prettyref{thm:non_splittable} and the definition of splittability.

\begin{corollary}
\label{cor:non_splittable}
For any positive integers $r,k$ with $r\geq 2,k\geq 3$, there exists an $k$-uniform hypergraph $H$ on $n(=r(k-1)+1)$ vertices, such that
$\phi_H\geq \frac{1}{k}$ and the induced simplicial complex $X$ is not $(\tau,r)$-splittable for all $\tau\in [-1,1]$.
\end{corollary}

We were also able to show that on the above example, for all $m,l$ such that $2\leq m<l\leq k$, the up-down walk from $X(m)$ to $X(l)$ has its top $r$ singular value as $1$ (for $r \approx n/k$).

\begin{theorem}
\label{thm:general_updown} 
For any positive integers $r,k$ with $r\geq 2,k\geq 3$, there exists an $k$-uniform hypergraph $H$ on $n(=r(k-1)+1)$ vertices such that $\mathsf{rank}_{\geq \tau}\paren{\mathsf{N}^{(2)}_{m,l}} \geq r$ for all $\tau \in [-1,1]$ but
$\phi_H\geq \frac{1}{k}$.
\end{theorem}

We use the following construction to show \prettyref{thm:non_splittable}, \prettyref{cor:non_splittable} and \prettyref{thm:general_updown}.

\begin{construction}
\label{cons:construction_2}
Take the vertex set of the hypergraph $H
(V,E)$ to be $V=[n]$ where $n=r(k-1)+1$ and the edge set $E=\set{e_1, e_2,\hdots,e_r}$ where $e_i=\set{0,k(i-1),\dots, ki-1}$.
Let $X$ be the simplicial complex induced by $H$ and $\mathsf{S}_{m,l},\mathsf{N}_{m,l}$ be the corresponding walk matrices.
\end{construction}

\begin{remark}
Remark 1.9 of \cite{lm16} considers all hypergraphs whose edges intersected in at most $k-2$ vertices to show a separation between co-boundary expansion and hypergraph expansion.
Here we consider a sub-class of such hypergraphs with edges intersecting in exactly 1 vertex.
Although the second singular value of the up-down walks and co-boundary expansion may seem related, a relation between them is not known.
Also, the way in which \cite{lm16} bounds the coboundary expansion is similar to how we bound the spectrum of the up-down walks.
However, here we also prove that the threshold rank (for any threshold) can be made arbitrarily large while having the same bound on the hypergraph expansion.
\end{remark}

First, we show that any swap-walk $\mathsf{S}_{l,k-l}$ has $\sigma_i=1$, for any $i \in [r]$.
\begin{lemma}
\label{lem:disconnected}
Given a hypergraph as per \prettyref{cons:construction_2}, 
we have that $\lambda_r(G_{1,k-1})=\sigma_r(\mathsf{S}_{1,k-1})= \sigma_r(\mathsf{S}_{k-1,1})= 1$.
\end{lemma}

\begin{proof}
Firstly using \prettyref{fact:bipartite_svalues} and eqn. \prettyref{eq:walkmatrix_swap} we have  $\lambda_i(G_{1,k-1})= \sigma_i(\mathsf{S}_{1,k-1}),\forall i \in [r]$.
We note that for any $i \in [r]$, the edge $\set{\set{k(i-1)},e_i\setminus\set{k(i-1)}}$ is the only edge in $G_{1,k-1}$ (and $G_{k-1,1}$) incident on the vertices $\set{k(i-1)},e_i\setminus\set{k(i-1)}$.
Again, $G_{1,k-1}$ has $r$ disconnected components, and hence $\lambda_r(G_{1,k-1})=\sigma_r(\mathsf{S}_{1,k-1}) = \sigma_r(\mathsf{S}_{k-1,1})=1$.
\end{proof}

\begin{lemma}
\label{lem:two_components}
Given a hypergraph as per \prettyref{cons:construction_2}, and for any $m,l\geq 2$ such that $m+l\leq k$, we have that $\lambda_r(G_{m,l})=\sigma_r(\mathsf{S}_{m,l})=1$.
\end{lemma}
\begin{proof}
Fix $m,l\geq 2$ with $m+l\leq k$.
Again, by \prettyref{fact:bipartite_svalues} and eqn. \prettyref{eq:walkmatrix_swap} we have  $\lambda_i(G_{m,l})= \sigma_i(\mathsf{S}_{m,l}), \forall i \in [r]$.
Now we claim that for each $i \in [r]$, the set of vertices given by $S_i=\set{s\in V(G_{m,l})|s\subseteq e_i}$ and $V(G_{m,l})\setminus S_i=\set{s\in V(G_{m,l})|s\subseteq E\setminus e_i}$ have no edges between them.
Fix any $s\in S_i$ and $t\in V(G_{m,l})\setminus S_i$.
Suppose there is a edge between $s \in S_i$ and $t \in V(G_{m,l}\setminus S_i)$.
Then $s\cup t\in X(m+l)$ which implies that  $s\cup t\subseteq e_j$ for some $j \in [r]$.
But, if $s\cup t\subseteq e_j$ for any $j \in [r]$ then we have $t\subseteq e_j \cap (E\setminus e_i)=\set{0}$ which contradicts $m,l\geq 2$.
Thus, for any $m,l\geq 2$ such that $m+l\leq k$, $G_{m,l}$ has disconnected components $\paren{S_i,V\setminus S_i}_{i=1}^{r}$, which implies, $\lambda_r(G_{m,l})=\sigma_r(\mathsf{S}_{m,l})=1$.
\end{proof}

\begin{lemma}
\label{lem:expanding_2}
Given a  hypergraph as per \prettyref{cons:construction_2} and an arbitrary set $S \subseteq V$ where $\mathsf{vol}_H(S) \leq \mathsf{vol}_H(V)/2$, we have that
$\phi_H(S) \geq \frac{1}{k}$.
\end{lemma}

\begin{proof}
We consider an arbitrary (non-empty) set $S \subset V$ such that $\mathsf{vol}_H(S) \leq \mathsf{vol}_H(V)/2$. Let $\abs{S \cap e_1}=t_1,\abs{S \cap e_2}=t_2,\dots,\abs{S \cap e_r}=t_r$ and let $t=t_1+t_2+\dots t_r$. 
We note that $\mathsf{vol}_H(V)=r(k-1)+r$ where $r(k-1)$ is the contribution from the vertices in $V \setminus \set{0} $ and we have a contribution of $r$ from the vertex $\set{0}$.
Next, we will compute the expansion $\phi_H(S)$ precisely. We will break into cases depending upon whether $\set{0} \in S$ or $\set{0} \notin S$.

First, consider the case where $\set{0} \in S$. We note that in this case $t_i \geq 1,\forall i \in [r]$.
In this case, we have that $\abs{\set{i|t_i = k}} < r/2$. 
This is because otherwise $\mathsf{vol}_H(S) \geq r+ \frac{r}{2}(k-1) > \frac{rk}{2}$ which contradicts $\mathsf{vol}_H(S) \leq \mathsf{vol}_H(V)/2$. Thus, $\abs{\set{i|t_i < k}} \geq r/2$ and hence $\partial_H(S) \geq r/2$. 
Next we have that $\mathsf{vol}_H(S)=r+\sum\limits_{i=1}^{r}(t_i-1)= t_1+t_2+\dots t_r=t$.
Using $\mathsf{vol}_H(S) \leq \mathsf{vol}_H(V)/2$, we have that $t \leq rk/2$ and we get
\begin{align*}
    \phi_H(S) = \frac{\abs{\partial_H(S)}}{\mathsf{vol}_H(S)} \geq \frac{r}{2t} \geq \frac{1}{k}.
\end{align*}

Next, we consider the case where $\set{0}  \notin S$. Let $t^+=\abs{\set{i}|t_i>0}$. Since $\set{0} \notin S$, we know that $t_i<k,\forall i \in [r]$ and hence the number of edges in boundary of $S$ is exactly $t^+$.
Moreover we can bound the volume of $S$ as $\mathsf{vol}_H(S) \leq t^+(k-1)$ and hence we have that,
\begin{align*}
    \phi_H(S) = \frac{\abs{\partial_H(S)}}{\mathsf{vol}_H(S)} \geq \frac{t^+}{t^+(k-1)} \geq \frac{1}{k}.
\end{align*}
\end{proof}

\begin{proof}[Proof of \prettyref{thm:non_splittable}]
Immediate from \prettyref{lem:expanding_2}, \prettyref{lem:two_components}, and \prettyref{lem:disconnected}.
\end{proof}

\begin{proof}[Proof of \prettyref{cor:non_splittable}]
Consider the hypergraph $H$ (and the induced simplicial complex) guaranteed by \prettyref{thm:non_splittable}.
Fix any $\tau \in [-1,1]$ and any $k$-splitting tree $\mathcal T$.
We note $G_{l,k-1}\in \mathsf{Swap}(\mathcal T, X)$ for some $l\in [k-1]$ as children of the root of $\mathcal T$ must be labeled $l$ and $k-l$ for some $l$.
Note that we have $\lambda_r(G_{l,k-l})=1$.
Hence, we have $\rank_{\geq \tau}(\mathsf{Swap}(\mathcal T, X))\geq \rank_{\geq\tau}(G_{l,k-l}) \geq r$.
Since, $\rank_{\geq \tau}(\mathsf{Swap}(\mathcal T, X))\geq r$ for any $k$-splitting tree $\mathcal T$, therefore 
$(X,\Pi)$ is not $(\tau, r)$-splittable for any $\tau \in [-1,1]$.
\end{proof}

\begin{lemma}
\label{lem:disconnected2}
Given a hypergraph as per \prettyref{cons:construction_2}, and any $m,l\in [k]$ such that $2\leq m\leq l$, we have that $\lambda_r(\mathsf{N}^{(2)}_{m,l})=1$.
\end{lemma}

\begin{proof}
Note that $\lambda_r(B_{2,k})=\lambda_r(\mathsf N_{2,k})= \sigma_r(\mathsf{D}_{2,k})$ where the first equality follows from \prettyref{fact:b_equals_n} and the last equality follows from \prettyref{fact:bipartite_svalues}.
We start by noting that for any $s\in X(2)$, is a subset of exactly one of $e_1,e_2,\hdots,e_r$. For the sake of contradiction, if $s$ is a subset of two of these, say $e_p$ and $e_q$ (for some $p,q \leq r$). 
Then we have $s\subseteq {e_p\cap e_q} = \set{0}$ and $\abs{s} =1$ which contradicts $s\in X(2)$.
Now, recall that $B_{2,k}$ is a bipartite graph with an edge between $s\in X(2)$ and $t\in X(k)$ iff $s\subseteq t$.
Hence, for any $s\in X(2)$, $s$ is adjacent to exactly one of the edges in $e_1,e_2,\dots,e_r$.
Now, it is easy to see that $B_{2,k}$ is a disconnected graph with $r$ connected components where for each $i=1,2,\hdots,r$ we have a component comprising of $e_i$ and its cardinality 2 subsets.
Therefore, $1=\lambda_r(B_{2,k})=\lambda_r(\mathsf N_{1,2})= \sigma_r(\mathsf{D}_{2,k})$.

Now, for any positive integers $m,l$ such that $2\leq m\leq l\leq k$ we can write $\mathsf{D}_{2,k}=\mathsf{D}_{2,m}\mathsf{D}_{m,l}\mathsf{D}_{l,k}$, and therefore by \prettyref{fact:svalue_ineq} we have,
\begin{align*}
    1=\sigma_r(\mathsf{D}_{2,k})\leq \sigma_1(\mathsf{D}_{2,m})\sigma_r(\mathsf{D}_{m,l}\mathsf{D}_{l,k})\leq \sigma_1(\mathsf{D}_{2,m})\sigma_r(\mathsf{D}_{m,l})\sigma_1(\mathsf{D}_{l,k})
    =\sigma_r(\mathsf{D}_{m,l})
\end{align*}
where the last inequality holds since $\sigma_1(\mathsf{D}_{2,m})=\sigma_1(\mathsf{D}_{l,k})=1$.
By \prettyref{fact:svalue_eigenvalue} we have $1=\sigma_r(\mathsf D_{m,l})=\lambda_r(\mathsf N^{(2)}_{m,l})$. Therefore, we have that $\mathsf{rank}_{\geq \tau}\paren{N_{m,l}^{(2)}} \geq r$, for all $\tau \in [-1,1]$
\end{proof}

\begin{proof}[Proof of \prettyref{thm:general_updown}]
    Immediate from \prettyref{lem:expanding_2} and \prettyref{lem:disconnected2}.
\end{proof}

\section{An expanding hypergraph with low link-expansion}
\label{sec:expanding_not_hdx}
In \prettyref{thm:non_hdx} we show that there is a family of expanding $k$-uniform hypergraphs $H$ with the induced simplicial complex having low link-expansion.

\begin{theorem}
\label{thm:non_hdx}
Let $n,k$ be any positive integers such that $n\geq 3k$ and $k\geq 3$, there exists a $k$-uniform hypergraph $H$ on $n+k-2$ vertices such that the link-expansion of the induced simplicial complex $X$ is at most $1-\cos{\frac{2\pi}{n}}$ and the expansion of $H$ is at least $\frac{1}{(3k)^k}$.
\end{theorem}

\prettyref{cons:construction} is a $k$-hypergraph with $n+k-2$ vertices such that its expansion is $\frac{1}{(3k)^k}$ while the link-expansion for the induced simplicial complex is $1-\cos{\frac{2\pi}{n}}$.

\begin{construction}
\label{cons:construction}
Take the vertex set of the hypergraph $H(V,E)$ to be $V=[n+k-2]$ and the edge set $E=\binom{[n]}{k}\cup \set{e\cup \set{n+1,\dots,n+k-2}|e\in C_n}$ where $C_n=\set{\set{i,i+1}|i\in [n]}\cup \set{\set{n,1}}$, i.e., $C_n$ is the set of edges in a cycle on $[n]$.
Let $X$ be the simplicial complex induced by $H$.
\end{construction}

The idea behind this construction is to have the cycle $C_n$ as the link of $\set{n+1,...,n+k-2}$.
As adding edges from $\binom{[n]}{k}$ does not change the fact that the link of $\set{n+1,\dots, n+k-2}$ is a cycle on $[n]$ we simply add all the edges from $\binom{[n]}{k}$.
We also remark that this approach may also produce other families of expanding hypergraph whose downward closure is not a high dimensional expander by replacing the cycle with some other non-expanding graph and $\binom{[n]}{k}$ with an expanding hypergraph.

\begin{remark}
In contrast with the construction in Remark 1.9 of \cite{lm16} which has edges intersecting in at most $k-2$ vertices, \prettyref{cons:construction} has many edges intersecting in $k-1$ vertices.
\end{remark}

\subsection{Expansion of $H$}
First, we show that $H$ has expansion at least $\frac{1}{(3k)^k}$.
To prove that the hypergraph as constructed above has expansion $\phi_H\geq \frac{1}{(3k)^k}$ we need the following facts about real numbers.
\begin{fact}
\label{fact:frac_ineq}
Let $a_1,a_2,b_1,b_2$ be positive real numbers, then 
\begin{align*}
    \frac{a_1+a_2}{b_1+b_2}\geq \min \set{\frac{a_1}{b_1},\frac{a_2}{b_2}}.
\end{align*}
\end{fact}
\begin{fact}
\label{fact:binom_approx}
Let $m\geq l$ be positive integers, then 
\begin{align*}
    \frac{m^l}{l^l} \leq \binom{m}{l}\leq m^l.
\end{align*}
\end{fact}

\begin{lemma}
\label{lem:exp}
For any $n,k$ such that $n \geq 3k$ and $k\geq 3$, the hypergraph $H$ as defined above has expansion $\phi_H\geq \frac{1}{(3k)^k}$.
\end{lemma}
\begin{proof}
Let $\emptyset\ne S\subset V$ such that $\vol_H(S)\leq \frac{\vol_H(V)}{2}$ and $\tau=\set{n+1,\dots,n+k-2}$.
Note that this implies $V\setminus S\ne \emptyset$.
In fact, for $n\geq 3k$, we can show that 
\begin{equation}
\label{eq:set_size}
    |S\setminus \tau|\leq \frac{2n}{3} \text{ and } |V\setminus (S\cup \tau)|\geq \frac{n}{3}
\end{equation}
by considering the following inequality,
\begin{align*}
    \abs{S\setminus \tau}\binom{n-1}{k-1}=\vol_H(S\setminus \tau)\leq \vol_H(S)\leq \frac{\vol_H(V)}{2}=\frac{1}{2}\left(n\paren{\binom{n-1}{k-1}+2}+(k-2)n\right),
\end{align*}
i.e., $\abs{S\setminus \tau}\leq \frac{n}{2}\paren{1+\frac{k}{\binom{n-1}{k-1}}}$.
Now, for $n\geq 3k$ and $k\geq 3$, we have that $1+\frac{k}{\binom{n-1}{k-1}}\leq 1+\frac{1}{3}$ and $\abs{S \setminus \tau} \leq \frac{2n}{3}$.

We now consider three cases according to whether (a) $S\cap \tau=\emptyset$, (b) $S\cap (V\setminus \tau)=\emptyset$, or (c) $S\cap \tau\ne \emptyset$ and $S\cap (V\setminus \tau)\ne \emptyset$ and analyze $\phi_H(S)$ in each of them as follows.
\begin{enumerate}[label=(\alph*)]
    \item Suppose $S\cap \tau=\emptyset$ then $\phi_H(S)=\frac{\abs{\partial_H(S)}}{\vol_H(S)}\geq \frac{\abs{\partial_H(S)\setminus \partial_H(\tau)}}{\vol_H(S)}$.
Firstly, $\vol_H(S)=|S|(\binom{n-1}{k-1}+2)$ as each vertex in $V\setminus \tau$ has degree $\binom{n-1}{k-1}+2$ where $\binom{n-1}{k-1}$ edges come from the set $\binom{[n]}{k}$ while two edges come from the edges constructed using $C_n$.
Now, $\partial_H(S)\setminus \partial_H(\tau)$ are the edges which are in the boundary of $S$ but don't have a vertex in $\tau$.
All such edges must belong to $\binom{[n]}{k}$.
Also, for a edge $e$ to be in $\partial_H(S)$, we must have $|e\cup S|\in \set{1,2,...,k-1}$ and there are exactly $\binom{\abs{S}}{l}\binom{n-\abs{S}}{k-l}$ edges so that their overlap with $S$ is exactly $l$, for $l\in[k-1]$.
Therefore, 
\begin{align*}
    \phi_H(S) \geq \frac{\sum_{l=1}^{k-1}\binom{\abs{S}}{l}\binom{n-\abs{S}}{k-l}}{|S|(\binom{n-1}{k-1}+2)}&\geq \frac{|S|\binom{n-\abs{S}}{k-1}}{\abs{S}(\binom{n-1}{k-1}+2)}\\
    &\geq \frac{\binom{n/3}{k-1}}{\binom{n-1}{k-1}+2}\geq \frac{\binom{n/3}{k-1}}{\binom{n-1}{k-1}} \numberthis \label{eq:complement}
\end{align*}
where the first inequality in eqn. \eqref{eq:complement} follows from $\abs{S}=\abs{S\setminus \tau}\leq \frac{2n}{3}$.
Now, using \prettyref{fact:binom_approx} we have $\phi_H(S) \geq \frac{\binom{n/3}{k-1}}{\binom{n-1}{k-1}} \geq \frac{1}{(3k)^k}$.

\item Suppose $S\cap (V\setminus \tau)=\emptyset$ then $\phi_H(S)=\frac{\abs{\partial_H(S)}}{\vol_H(S)}=\frac{n}{\abs{S}n}\geq \frac{1}{k}$. Here $\vol_H(S)=\abs{S}n$ as the degree of each vertex in $\tau$ is $n$. 
As $S\subset \tau$ all the edge constructed from $C_n$ cross $S$ while no edge from $\binom{[n]}{k}$ can cross $S$, therefore $|\partial_H(S)|=n$.

\item Suppose $S\cap \tau \ne \emptyset$ and $S\cap (V\setminus \tau)\ne \emptyset$.
Note that we can write $\vol_H(S)$ as $\vol_H(S)=\vol_H(S\cap \tau)+\vol_H(S\setminus \tau)$.
Also note, $\partial_H(S)= E_1 \sqcup E_2$ where $E_1=\partial_H(S)\cap \partial_H(\tau)$, i.e., edges between $V\setminus (S\cup \tau)$ and $S\cap \tau$ and $E_2=\partial_H(S)\setminus \partial_H(\tau)$, i.e., the edges between $V\setminus S$ and $S\setminus \tau$.
Now, we use \prettyref{fact:frac_ineq} to break $\phi_H(S)$ into two terms as follows,
\begin{align*}
    \phi_H(S)=\frac{|\partial_H(S)|}{\vol_H(S)}
    &= \frac{\abs{E_1}+\abs{E_2}}{\vol_H(S\cap \tau)+\vol_H (S\setminus \tau))}\\
    &\geq \min \set{\frac{\abs{E_1}}{\vol_H(S\cap \tau)},\frac{\abs{E_2}}{\vol_H (S\setminus \tau))}}\numberthis \label{eq:eq_4}.
\end{align*}
We can again use arguments similar to cases (a) and (b) to lower bound the two arguments. To bound the first argument of $\min$ in eqn. \prettyref{eq:eq_4}, we see that $\vol_H(S\cap \tau)=\abs{S\cap \tau}n\leq kn$ as the edges constructed from $C_n$ are precisely the edges incident on vertices in $\tau$.
Now, take any vertex $i\in V\setminus(S\cup \tau)$.
Let $j=i+1$ if $i\ne n$, otherwise let $j=1$.
Then we have an edge in the hypergraph of the form $\set{i,j}\cup \set{n+1,\dots, n+k-2}$ which clearly has a non-zero intersection with $S\cap \tau$.
Hence, $\abs{E_1}\geq \abs{V\setminus (S\cup \tau)}$.
Using \prettyref{eq:set_size} we have $\abs{V\setminus (S\cup \tau)}\geq \frac{n}{3}$.
This means $\frac{\abs{E_1}}{\vol_H(S\cap \tau)}\geq \frac{1}{3k}$.

Now consider the second argument of $\min$ in eqn. \prettyref{eq:eq_4}.
Again, as each vertex in $V\setminus \tau$ has degree $\binom{n-1}{k-1}+2$ where $\binom{n-1}{k-1}$ edges come from the set $\binom{[n]}{k}$ while two edges come from the edges constructed using $C_n$, we have $\vol_H(S\setminus \tau)=\abs{S\setminus \tau}(\binom{n-1}{k-1}+2)$.
To lower bound the numerator, we consider edges with exactly one vertex in $S\setminus \tau$ and $k-1$ vertices in $V\setminus (S\cup \tau)$, i.e.,
\begin{align*}
    \abs{E_2}\geq \abs{S\setminus \tau}\binom{n-\abs{S\setminus \tau}}{k-1}\geq \abs{S\setminus \tau}\binom{\frac{n}{3}}{k-1}
\end{align*} where the final inequality uses \prettyref{eq:set_size}.
Thus we obtain $\frac{\abs{E_2}}{\vol_H(S\setminus \tau)}\geq \frac{\binom{n/3}{k-1}}{\binom{n-1}{k-1}}\geq \frac{1}{(3k)^k}$ where the last inequality follows from \prettyref{fact:binom_approx}.

Finally, taking the minimum of the two bounds gives us, $\phi_H(S)\geq \frac{1}{(3k)^k}$.
\end{enumerate}

Thus, $\phi_H=\min_{S|\vol_H(S)\leq \vol_H(V)/2}\phi_H(S)\geq \frac{1}{(3k)^k}$.
\end{proof}

\subsection{$X$ is not a high-dimensional expander}
We show that the simplicial complex $X$ is not a $\gamma$-HDX (refer \prettyref{def:local_hdx}).
For this we consider the face $\tau=\set{n+1,n+2,\hdots,n+k-2}$
and the link complex $X_{\tau}$ corresponding to this face $\tau$, defined as
\begin{align*}
    X_{\tau} = \set{s \setminus \tau | (\tau \subseteq s) \land (s \in X) } \mper
\end{align*}
We note that by definition of $X_\tau$ and our construction in \prettyref{cons:construction}, the two-dimensional link complex $X_\tau$ here is the downward closure of $C_n$.
Hence, the corresponding skeleton graph $G(X_\tau)$ is the cycle on $[n]$.

\begin{fact}[Folklore]
The second singular value of the normalized adjacency matrix of a $n$-cycle is $\cos{\frac{2\pi}{n}}$.
\end{fact}
Therefore, by \prettyref{def:local_hdx} we have the following lemma.
\begin{lemma}
\label{lem:low_link}
    $X$ has link-expansion at most $1-\cos{\frac{2\pi}{n}}$.
\end{lemma}
The \prettyref{thm:non_hdx} follows directly from \prettyref{lem:low_link} and \prettyref{lem:exp}.

%% file: appendix.tex
\input{omitted}

%% file: omitted.tex
\section{Omitted Proofs}

\subsection{Proof of \prettyref{fact:bipartite_svalues}}
\label{app:proof_of_bipartite_svalues}

\begin{proof}[Proof of \prettyref{fact:bipartite_svalues}]
We claim that for a positive number $\alpha$ the following are equivalent:
\begin{enumerate}[label=(\alph*)]
    \item $\alpha$ is a singular value of $\mathsf A$;
    \item $\alpha$ is an eigenvalue of $\mathsf B$;
    \item $-\alpha$ is an eigenvalue of $\mathsf B$.
\end{enumerate}
\begin{description}
    \item[(a)$\Rightarrow$(b)] Suppose $\alpha$ is a singular value of $\mathsf A$ with $u$, $v$ as left and right singular vectors, respectively.
Then,
\begin{align*}
\mathsf B\begin{bmatrix}u\\ v \end{bmatrix}
=\begin{bmatrix}
0 && \mathsf A\\
\mathsf{A}^{\dagger} && 0
\end{bmatrix}\begin{bmatrix}u\\ v \end{bmatrix}=\begin{bmatrix}\mathsf Av\\ \mathsf A^\dagger u \end{bmatrix}=\alpha \begin{bmatrix}u\\ v\end{bmatrix},
\end{align*}
i.e., $\alpha$ is an eigenvalue.
    \item[(b)$\Rightarrow$(c)] Suppose $\alpha$ is an eigenvalue of $\mathsf B$.
    Then there exists $u\in W$ and $v\in V$ with at least one of them is non-zero such that 
    \begin{equation*}
        \alpha\begin{bmatrix}u\\ v \end{bmatrix}=\mathsf B\begin{bmatrix}u\\ v \end{bmatrix}=\begin{bmatrix}\mathsf Av\\ \mathsf A^\dagger u \end{bmatrix}.
    \end{equation*}
    Therefore, $\mathsf Av=\alpha u$ and $\mathsf A^\dagger u= \alpha v$ which implies $u$ and $v$ both must be non-zero.
    Also note that 
    \begin{equation*}
        -\alpha\begin{bmatrix}u\\ -v \end{bmatrix}=\begin{bmatrix} -\mathsf Av\\  \mathsf A^\dagger u \end{bmatrix}=\mathsf B\begin{bmatrix}u\\ -v \end{bmatrix}.
    \end{equation*}
    Hence, $-\alpha$ is an eigenvalue.
    \item[(c)$\Rightarrow$(a)] Suppose $-\alpha$ is an eigenvalue of $\mathsf B$.
    Then there exists $u\in W$ and $v\in V$ with at least one of them is non-zero such that 
    \begin{equation*}
        -\alpha\begin{bmatrix}u\\ v \end{bmatrix}=\mathsf B\begin{bmatrix}u\\ v \end{bmatrix}=\begin{bmatrix}\mathsf Av\\ \mathsf A^\dagger u \end{bmatrix}.
    \end{equation*}
    Therefore, $\mathsf Av= -\alpha u$ and $\mathsf A^\dagger u= -\alpha v$ which implies $u$ and $v$ both must be non-zero.
    Hence, $\mathsf Av= \alpha (-u)$ and $\mathsf A^\dagger (-u)= \alpha v$, i.e., $\alpha$ is a singular value.
\end{description}
The statement now follows from the above claim and the fact that $\rank(\mathsf A)$ is equal to number of positive singular values of $\mathsf A$.
\end{proof}

\subsection{Proof of \prettyref{lem:pi_level_relation}}
\label{app:proof_of_pi_level_relation}
\begin{proof}[Proof of \prettyref{lem:pi_level_relation}]
Fix $m\leq l$ and $s\in X(m)$.
Then we have
\begin{align*}
\sum_{t\in X(l)|t\supseteq s}\Pi_l(t)&=\frac{1}{\binom{k}{l}}\sum_{t\in X(l)|t\supseteq s}\sum_{t'\in X(k)|t'\supseteq t}\Pi_k(t')&&\text{(definition of }\Pi_l)\\
&=\frac{1}{\binom{k}{l}}\sum_{t'\in X(k)|t'\supseteq s}\sum_{t\in X(l)|t'\supseteq t\supseteq s}\Pi_k(t')&&\text{(Rearranging the sum)}\\
&=\frac{1}{\binom{k}{l}}\sum_{t'\in X(k)|t'\supseteq s}\binom{k-m}{l-m}\Pi_k(t')&&\paren{\text{Summing }\binom{k-m}{l-m}\text{ constant terms}}\\
&=\frac{\binom{l}{m}}{\binom{k}{m}}\sum_{t'\in X(k)|t'\supseteq s}\Pi_k(t')&&\paren{\text{since }\frac{\binom{k-m}{l-m}}{\binom{k}{l}}=\frac{\binom{l}{m}}{\binom{k}{m}}}\\
&=\binom{l}{m}\Pi_m(s)&&\text{(definition of }\Pi_m).
\end{align*}
\end{proof}

\subsection{Proof of \prettyref{lem:updown_adj}}
\label{app:proof_of_updown_adj}
\begin{proof}[Proof of \prettyref{lem:updown_adj}]
For any $f\in \R^{X(i)}$ and $g\in \R^{X(i+1)}$,
\begin{align*}
\inprod{\mathsf U_if,g}=\E_{s\sim \Pi_{i+1}}[\mathsf U_if](s)g(s)&=\E_{s\sim \Pi_{i+1}}\left[\E_{s'\in X(i)|s'\subset s}f(s')\right]g(s)\\
&=\E_{s\sim \Pi_{i+1}}\left[\E_{s'\sim \Pi_i|s'\subset s}f(s')\right]g(s)\\
&=\E_{\substack{(s,s')\sim (\Pi_{i+1},\Pi_i)\\s \supseteq s'}}f(s')g(s),
\end{align*}
while
\begin{align*}
\inprod{f,\mathsf D_{i+1}g}=\E_{s'\sim \Pi_{i}}f(s')[\mathsf D_{i+1}g](s')
&=\E_{s'\sim \Pi_{i}}f(s')\left[\E_{s\sim \Pi_{i+1}}g(s)\right]\\
&=\E_{\substack{(s,s')\sim (\Pi_{i+1},\Pi_i)\\s \supseteq s'}}f(s')g(s).
\end{align*}
Hence, for any $f\in\R^{X(i)}$ and $g\in \R^{X(i+1)}$ we have,
\begin{align*}
\inprod{\mathsf U_if,g}=\E_{\substack{(s,s')\sim (\Pi_{i+1},\Pi_i)\\s \supseteq s'}}f(s')g(s)=\inprod{f,\mathsf D_{i+1}g},
\end{align*}
i.e., $\mathsf D_{i+1}=\mathsf U_i^\dagger$.
\end{proof}

\subsection{Proof of \prettyref{lem:swap_adj}}
\label{app:proof_of_swap_adj}
\begin{proof}[Proof of \prettyref{lem:swap_adj}]
For any $f\in \R^{X(l)}$ and $g\in \R^{X(m)}$,
\begin{align}
\inprod{\mathsf S_{m,l}f,g}=\E_{s\sim \Pi_{m}}[\mathsf S_{m,l}f](s)g(s)
&=\E_{s\sim \Pi_{m}}\left[\E_{s'\sim \Pi_{m+l}|s'\subset s}f(s'\setminus s)\right]g(s) \nonumber \\
&=\E_{\substack{(s,s')\sim (\Pi_{m},\Pi_{m+l})\\s \subseteq s'}}f(s'\setminus s)g(s),\label{eq:swap_ls}
\end{align}
while
\begin{align}
\inprod{f,\mathsf S_{l,m}g}=\E_{t\sim \Pi_{l}}f(t)[\mathsf S_{m,l}g](s)
&=\E_{t\sim \Pi_{l}}f(t)\left[\E_{s'\sim \Pi_{m+l}|s'\subset t}g(s'\setminus t)\right] \nonumber\\
&=\E_{\substack{(t,s')\sim (\Pi_{l},\Pi_{m+l})\\t \subseteq s'}}f(t)g(s'\setminus t).\label{eq:swap_rs}
\end{align}
Note that for any  $\sigma'\in X(m+l)$, $\sigma\in X(m)$, and $\tau\in X(l)$ such that $\sigma,\tau\subseteq \sigma'$ we have
\begin{align*}
\pr{\substack{(s,s')\sim (\Pi_{m},\Pi_{m+l})\\s \subseteq s'}}{s=\sigma \text{ and }s'=\sigma'}&=\pr{s'\sim \Pi_{m+l}}{s'=\sigma'}\pr{s\sim \Pi_{m}|s\subseteq \sigma'}{s=\sigma}\\
&=\Pi_{m+l}(\sigma') \frac{1}{\binom{m+l}{m}}
\end{align*}
and similarly,
\begin{align*}
\pr{\substack{(t,s')\sim (\Pi_{l},\Pi_{m+l})\\t \subseteq s'}}{t=\tau \text{ and }s'=\sigma'}=\Pi_{m+l}(\sigma') \frac{1}{\binom{m+l}{l}}=\Pi_{m+l}(\sigma') \frac{1}{\binom{m+l}{m}},
\end{align*}
i.e.,
\begin{align}
\pr{\substack{(t,s')\sim (\Pi_{l},\Pi_{m+l})\\t \subseteq s'}}{t=\tau \text{ and }s'=\sigma'}=\pr{\substack{(s,s')\sim (\Pi_{m},\Pi_{m+l})\\s \subseteq s'}}{s=\sigma \text{ and }s'=\sigma'}.\label{eq:prob_eq}
\end{align}
Now using \prettyref{eq:prob_eq}, we can equate the RHS of \prettyref{eq:swap_ls} and \prettyref{eq:swap_rs} by substituting $t=s'\setminus s$, i.e., we have

\begin{align}
\E_{\substack{(s,s')\sim (\Pi_{m},\Pi_{m+l})\\s \subseteq s'}}f(s'\setminus s)g(s)=\E_{\substack{(t,s')\sim (\Pi_{l},\Pi_{m+l})\\t \subseteq s'}}f(t)g(s'\setminus t).\label{eq:exp_eq}
\end{align}
Hence, combining \prettyref{eq:swap_ls}, \prettyref{eq:swap_rs} and \prettyref{eq:exp_eq} we have that for any $f\in\R^{X(i)}$ and $g\in \R^{X(i+1)}$,
\begin{align*}
\inprod{\mathsf S_{m,l},g}=\E_{\substack{(s,s')\sim (\Pi_{m},\Pi_{m+l})\\s \subseteq s'}}f(s'\setminus s)g(s)=\E_{\substack{(t,s')\sim (\Pi_{l},\Pi_{m+l})\\t \subseteq s'}}f(t)g(s'\setminus t)=\inprod{f,\mathsf S_{l,m}g},
\end{align*}
i.e., $\mathsf S_{m,l}^\dagger=\mathsf S_{l,m}$.
\end{proof}

\subsection{Proof of \prettyref{lem:uniform_count2}}
\label{app:proof_of_uniform_count2}

\begin{proof}[Proof of \prettyref{lem:uniform_count2}]
Fix $m\leq l$.
For any $s\in X(m)$,
\begin{align*}
w(s,t)&=\binom{k}{l}\sum_{\substack{s'\in X(l)\\ s'\supseteq s\cup t}}\Pi_{l}(s') &&\text{(by definition of }w(s,t))\\
&=\frac{\binom{k}{l}}{\binom{k}{l}}\sum_{\substack{s'\in X(l)\\ s'\supseteq s\cup t}}\sum_{e\in E|s'\subseteq e}\Pi_k(e)&&\text{(by definition of }\Pi_{l})\\
&=\sum_{e\in E|s\cup t\subseteq e}\sum_{\substack{s'\in X(l)\\ e\supseteq s'\supseteq s\cup t}}\Pi_k(e)&&\text{(Rearranging the sum)}\\
&=\binom{k-\abs{s\cup t}}{l-\abs{s\cup t}}\sum_{e\in E|s\cup t\subseteq e}\Pi_k(e)&&\text{(Summing }\binom{k-\abs{s\cup t}}{l-\abs{s\cup t}}\text{ constant terms)}\\
\end{align*}
and
\begin{align*}
    \mathsf{deg}_{B^{(2)}_{m,l}}(s)&=\binom{k}{l}\sum_{t'\in X(m)}\sum_{s'\in X(l)|s'\supseteq s\cup t'}\Pi_l(s')\\
    &=\binom{k}{l}\sum_{s'\in X(l)|s'\supseteq s}\sum_{t'\in X(m)|t'\subseteq s'}\Pi_l(s')&&\text{(Rearranging the sum)}\\
    &=\binom{k}{l}\sum_{s'\in X(l)|s'\supseteq s}\binom{l}{m}\Pi_l(s')&&\text{(Summing }\binom{l}{m}\text{ constant terms)}\\
    &=\binom{k}{l}{\binom{l}{m}}^2\Pi_m(s)&&\text{(by \prettyref{lem:pi_level_relation})}\\
    &={\binom{l}{m}}^2\frac{\binom{k}{l}}{\binom{k}{m}}\sum_{e\in E|e\supseteq s}\Pi_k(e).&&\text{(Definition of }\Pi_m)%
\end{align*}
\end{proof}

\subsection{Proof of \prettyref{fact:b_equals_n}}
\label{app:proof_of_b_equals_n}

\begin{proof}[Proof of \prettyref{fact:b_equals_n}]
Fix $m\leq l$.
Let $\mathsf W$ denote the transition matrix of the random walk on $B_{m,l}$.
Let $\mathsf A$ and $\mathsf D$ denote the (weighted) adjacency matrix and the degree matrix for $B_{m,l}$, respectively.
Then $\mathsf {W=D^{-1}A}$.
Given a $t\in X(m)\cup X(l)$ let $\mathsf e_{t}\in \R^{X(m)\cup X(l)}$ denote the function which is equal to 1 at $t$ and 0 otherwise.
Note that $\set{\mathsf e_{s}|s\in X(m)\cup X(l)}$ forms a basis for $\R^{X(m)\cup X(l)}$.
For any $s, t\in X(m)\cup X(l)$,
\begin{equation}
\label{eq:b_walk_matrix}
[\mathsf W\mathsf e_t](s)=\mathsf W(s,t)=[\mathsf D^{-1}\mathsf A](s,t)=\frac{w(s,t)}{\mathsf{deg}_{B_{m,l}}(s)}
\end{equation}
If $s\in X(m)$ and $t\in X(m)$ we have $w(s,t)=0$, therefore, $[\mathsf W\mathsf e_t](s)=0$.
While $[\mathsf N_{m,l} \mathsf e_t](s)=[\mathsf U_{l,m}\mathsf e_t](s)=0$.
Hence, if $s\in X(m)$ and $t\in X(m)$ we have $[\mathsf N_{m,l} \mathsf e_t](s)=[\mathsf W\mathsf e_t](s)=0$.
Similarly, if $s\in X(l)$ and $t\in X(l)$ we have $[\mathsf N_{m,l} \mathsf e_t](s)=[\mathsf W\mathsf e_t](s)=0$.

If $s\in X(m)$ and $t\in X(l)$ then 
\begin{align*}
    \mathsf{deg}_{B_{m,l}}(s)&=\binom{k}{l}\sum_{t'\in X(l)|t'\supseteq s}\Pi_l(t')\\
    &=\binom{k}{l}\binom{l}{m}\Pi_m(s),
\end{align*}
and $w(s,t)=\binom{k}{l}\Pi_l(t)$, therefore, by \prettyref{eq:b_walk_matrix} we have $[\mathsf W\mathsf e_t](s)=\frac{\Pi_l(t)}{\binom{l}{m}\Pi_m(s)}$,
while $[\mathsf N_{m,l} \mathsf e_t](s)=[\mathsf D_{m,l}\mathsf e_t](s)=\frac{\Pi_l(t)}{\binom{l}{m}\Pi_m(s)}$, i.e., $[\mathsf W\mathsf e_t](s)=[\mathsf N_{m,l} \mathsf e_t](s)$.

Finally, if $s\in X(l)$ and $t\in X(m)$ then 
\begin{align*}
    \mathsf{deg}_{B_{m,l}}(s)&=\binom{k}{l}\sum_{t'\in X(m)|t'\subseteq s}\Pi_l(s)\\
    &=\binom{k}{l}\binom{l}{m}\Pi_l(s),
\end{align*}
and $w(s,t)=\binom{k}{l}\Pi_l(s)$, therefore, by \prettyref{eq:b_walk_matrix} we have $[\mathsf W\mathsf e_t](s)=\frac{1}{\binom{l}{m}}$,
while $[\mathsf N_{m,l} \mathsf e_t](s)=[\mathsf D_{m,l}\mathsf e_t](s)=\frac{\Pi_l(t)}{\binom{l}{m}\Pi_m(s)}$, i.e., $[\mathsf W\mathsf e_t](s)=[\mathsf N_{m,l} \mathsf e_t](s)$.
\end{proof}

\subsection{Proof of \prettyref{fact:b2_equals_updown}}
\label{app:proof_of_b2_equals_updown}

\begin{proof}[Proof of \prettyref{fact:b2_equals_updown}]
Fix $m\leq l$.
Let $\mathsf W$ denote the transition matrix of the random walk on $B^{(2)}_{m,l}$.
Let $\mathsf A$ and $\mathsf D$ denote the (weighted) adjacency matrix and the degree matrix for $B^{(2)}_{m,l}$, respectively.
Then $\mathsf {W=D^{-1}A}$.
Given a $t\in X(m)$ let $\mathsf e_{t}\in \R^{X(m)}$ denote the function which is equal to 1 at $t$ and 0 otherwise.
Note that $\set{\mathsf e_{s}|s\in X(m)}$ forms a basis for $\R^{X(m)}$.
For any $s, t\in X(m)$,
\begin{equation}
\label{eq:b2_walk_matrix}
[\mathsf W\mathsf e_t](s)=\mathsf W(s,t)=[\mathsf D^{-1}\mathsf A](s,t)=\frac{w(s,t)}{\mathsf{deg}_{B^{(2)}_{m,l}}(s)}
\end{equation}
Now, 
\begin{align*}
    \mathsf{deg}_{B^{(2)}_{m,l}}(s)&=\binom{k}{l}\sum_{t'\in X(m)}\sum_{s'\in X(m)|s'\supseteq s\cup t'}\Pi_l(t')\\
    &=\binom{k}{l}\sum_{s'\in X(l)|s'\supseteq s}\sum_{t'\in X(m)|t'\subseteq s'}\Pi_l(t')&&\text{(Rearranging the sum)}\\
    &=\binom{k}{l}\sum_{s'\in X(l)|s'\supseteq s}\binom{l}{m}\Pi_l(t')&&\text{(Summing }\binom{l}{m}\text{ constant terms)}\\
    &=\binom{k}{l}{\binom{l}{m}}^2\Pi_m(s),&&\text{(by \prettyref{lem:pi_level_relation})}
\end{align*}
and 
\begin{align*}
w(s,t)=\binom{k}{l}\sum_{s'\supseteq s\cup t}\Pi_l(s')&=\binom{k}{l}\binom{l}{m}\Pi_m(s)\sum_{s'\supseteq s}\frac{\Pi_l(s')\mathbb{1}_{s'\supseteq t}}{\binom{l}{m}\Pi_m(s)}\\
&=\binom{k}{l}\binom{l}{m}\Pi_m(s)\E_{\substack{s'\sim \Pi_l\\s'\supseteq s}}\mathbb{1}_{s'\supseteq t},
\end{align*}
therefore, by \prettyref{eq:b2_walk_matrix} we have
\begin{align*}
[\mathsf W\mathsf e_t](s)=\E_{\substack{s'\sim \Pi_l\\s'\supseteq s}}\frac{\mathbb{1}_{s'\supseteq t}}{\binom{l}{m}},
\end{align*}
while 
\begin{align*}
[\mathsf D_{m,l}\mathsf U_{l,m} \mathsf e_t](s)&=\E_{\substack{s'\sim \Pi_l\\ s'\supseteq s}}[\mathsf U_{l,m}\mathsf e_t](s')&&\text{(Defintion of }\mathsf D_{m,l})\\
&=\E_{\substack{s'\sim \Pi_l\\ s'\supseteq s}}\E_{\substack{s''\in X(m)\\s''\subseteq s'}}\mathsf e_t(s'')&&\text{(Definition of }\mathsf U_{l,m})\\
&=\E_{\substack{s'\sim \Pi_l\\s'\supseteq s}}\frac{\mathbb{1}_{s'\supseteq t}}{\binom{l}{m}},&&(\mathsf e_t(s'')=1\text{ iff } s''=t) 
\end{align*}
i.e., $[\mathsf W\mathsf e_t](s)=[\mathsf D_{m,l}\mathsf U_{l,m} \mathsf e_t](s)$. Hence, $\mathsf W=\mathsf D_{m,l}\mathsf U_{l,m}$.
\end{proof}

\subsection{Proof of \prettyref{fact:svalue_ineq}}
\label{app:proof_of_svd_product}

\begin{proof}[Proof of \prettyref{fact:svalue_ineq}.]\footnote{The proof presented here is adapted from \url{https://math.stackexchange.com/questions/3892626/singular-values-of-product-of-matrices/} \cite{Gro20}.}
Using the Courant-Fischer theorem for singular values\footnote{This is a simple corollary of the usual variant of Courant-Fisher theorem which follows from \prettyref{fact:svalue_eigenvalue}, i.e., $\sigma_i(\mathsf{A})=\sqrt{\lambda_i(\mathsf{A}^\dagger \mathsf{A})}$.} we have that,
\begin{align*}
    \sigma_i(\mathsf{A}\mathsf{B})&= \sqrt{\lambda_i((\mathsf{A}\mathsf{B})^\dagger \mathsf{A}\mathsf{B})} =  \max_{\substack{S \subseteq \mathbb{R}^p\\
    \mathsf{dim}(S) = i}}\min_{\substack{\mathbf{x} \in S\\\norm{{x}}=1}}\norm{\mathsf{A}\mathsf{B}{x}}\\
                &\leq \max_{\substack{S \subseteq \mathbb{R}^p\\\mathsf{dim}(S) = i}}\min_{\substack{{x} \in S\\\norm{{x}}=1}}\norm{\mathsf{A}}\norm{\mathsf{B}{x}}\\
                &\leq \sigma_1(\mathsf{A})\max_{\substack{S \subseteq \mathbb{R}^p\\\mathsf{dim}(S) = i}}\min_{\substack{{x} \in S\\\norm{{x}}=1}}\norm{\mathsf{B}{x}}\\
                &= \sigma_1(\mathsf{A})\sigma_i(\mathsf{B}).
\end{align*}
Now, $\sigma_i(\mathsf{A}\mathsf{B})=\sigma_i(\mathsf{B}^\dagger \mathsf{A}^\dagger)$ and by the inequality shown above $\sigma_i(\mathsf{B}^\dagger \mathsf{A}^\dagger)\leq \sigma_1(\mathsf{B}^\dagger)\sigma_i(\mathsf{A}^\dagger)=\sigma_i(\mathsf{A})\sigma_1(\mathsf{B})$.
\end{proof}